\documentclass[12pt,reqno]{amsart}

\usepackage{amscd,amssymb,amsthm}
\usepackage{graphicx}
\usepackage{multirow}
\usepackage{array}
\usepackage{color}
\usepackage{bbm}
\usepackage{epstopdf}

\newcommand{\RNum}[1]{\uppercase\expandafter{\romannumeral #1\relax}}

\newtheorem{theorem}{Theorem}[section]
\newtheorem{corollary}{Corollary}[section]
\newtheorem{proposition}{Proposition}[section]
\newtheorem{lemma}{Lemma}[section]
\newtheorem{example}{Example}[section]
\newtheorem{definition}{Definition}[section]

\numberwithin{equation}{section}

\def\asubs#1{{
\vbox{\hrule height .2pt \kern 1pt \hbox{$\scriptstyle {#1}\kern
1pt$} }\kern-.05pt \vrule width .2pt }}
\def\hang{\hangindent=\parindent\noindent}
\def\urltilda{\kern -.15em\lower .7ex\hbox{\~{}}\kern .04em}
\begin{document}
\def\hang{\hangindent=\parindent\noindent}

\thispagestyle{empty}

\vspace*{1mm}

\begin{center}
\MakeUppercase{\bf {A form of multivariate Pareto distribution with
applications to financial risk measurement}}

\bigskip

\textsc{Jianxi Su and
Edward Furman\footnote{Corresponding author. Department of Mathematics and Statistics, York University, Toronto, Ontario M3J 1P3, Canada; E-mail: efurman@mathstat.yorku.ca; Tel: 416-736-2100 Ext 33768; Fax: 416-736-5757.}}
\smallskip

\textit{Department of Mathematics and Statistics, York University, \\
Toronto, Ontario M3J 1P3, \\ Canada}

\bigskip
\end{center}

\begin{quote}
\textbf{Abstract.} A new multivariate distribution possessing arbitrarily
parame-\\ trized and positively dependent univariate Pareto margins is introduced. Unlike the probability law
of Asimit et al. (2010)
[Asimit, V., Furman, E. and Vernic, R. (2010) On a multivariate Pareto distribution.
\textit{Insurance: Mathematics and Economics}
\textbf{46}(2), 308 -- 316], the structure in this paper is absolutely continuous
with respect to the corresponding Lebesgue measure. The distribution is
of importance to actuaries through its connections to the popular frailty models,
as well as because of the capacity to describe dependent heavy-tailed risks. The genesis of
the new distribution is linked to a number of existing probability models, and useful
characteristic results are proved. Expressions for, e.g., the decumulative
distribution and probability density functions, (joint) moments and regressions are developed. The distributions
of minima and maxima, as well as, some weighted risk measures are employed to exemplify
possible applications of the distribution in insurance.

\bigskip \noindent\textit{Keywords and phrases: }
Multivariate Pareto distributions, characterizations, dependence, weighted risk measures,
minima, maxima.

%\bigskip \noindent\textit{Mathematics Subject Classification}: IM01, IM10, IM54.

\end{quote}

\newpage

\section{introduction}
At the outset, we fix the probability space $(\Omega,\Sigma,\mathbf{P})$ and
define the random vector (r.v.) $\mathcal{X}\ni\mathbf{X}=(X_1,\ldots,X_n)'$ as a map from
$(\Omega,\Sigma)$ to the $n(\in\mathbf{N})$-dimensional Borel space
$(\mathbf{R}_+^n,\mathcal{B}(\mathbf{R}_+^n))$. The cumulative distribution
function (c.d.f.) of $\mathbf{X}$ is in the sequel denoted by
$F_{1,\ldots,n}(x_1,\ldots,x_n):=\mathbf{P}[X_1\leq x_1,\ldots,X_n\leq x_n]$, and
the corresponding probability density function (p.d.f.)  by
$f_{1,\ldots,n}(x_1,\ldots,x_n):=\partial^n/(\partial x_1\cdots \partial x_n)
F_{1,\ldots,n}(x_1,\ldots,x_n)$, where $(x_1,\ldots,x_n)'\in\mathbf{R}_{+}^n:=(0,\ \infty)^n$.
Finally, $F_i(x)$ and $f_i(x)$ denote, respectively,
the marginal c.d.f. and p.d.f. of $X_i,\ i=1,\ldots,n$.
Clearly, when the coordinates of $\mathbf{X}$ are stochastically independent, then
there is only one way to formulate the c.d.f. $F_{1,\ldots,n}$, whereas the shapes
of the just-mentioned c.d.f. are infinite otherwise.

We further discuss the so-called multivariate reduction approach to creating random
vectors with dependent coordinates. This paves the way to introducing the main object of
our interest in Section \ref{sec2}.
For applications of the multivariate reduction method in insurance, we refer to, e.g.,
Vernic (1997, 2000), Pfeifer and Ne\v slehov\' a (2004), Furman and Landsman (2005, 2010),
Boucher et al. (2008) and Tsanakas (2008), as well as to the references therein.

Let $\mathcal{Y}\ni\mathbf{Y}=(Y_1,\ldots,Y_{n+1})'$ be an $(n+1)$-variate r.v. with
mutually independent univariate margins distributed gamma. Namely, for
$j=1,\ldots,n+1$, the p.d.f. of
$Y_j\backsim Ga(\gamma_j(\in\mathbf{R}_+),\alpha_j(\in\mathbf{R}_+))$ is given by
\begin{equation}
\label{Gampdf}
g(y;\gamma_j,\alpha_j)=e^{-\alpha_j y}\frac{y^{\gamma_j-1}\alpha_j^{\gamma_j}}
{\Gamma(\gamma_j)},\ y\in\mathbf{R}_+
\end{equation}
with the corresponding Laplace transform being well-defined on $\mathbf{R}_{+}:=(0,\ \infty)$ (the interval of
interest herein) and
given by
\begin{equation}
\label{GamLS}
\hat{G}(x;\gamma_j,\alpha_j)=\int_0^\infty e^{-xy}g(y)dy=\left(
1+\frac{x}{\alpha_j}
\right)^{-\gamma_j}.
\end{equation}

\begin{definition}[Furman, 2008; Furman and Landsman, 2010]
Let
$A\in Mat_{n\times (n+1)}(\mathbf{R}_{0,+})$ denote a deterministic $n\times (n+1)$ matrix
with  suitable non-negative entries. Then $\mathbf{X}=A\mathbf{Y}$ is distributed
 $n$-variate gamma
with shape parameters $\gamma_i^\ast=t_i(\gamma_1,\ldots,\gamma_{n+1})$
and rate parameters
$\alpha_i^\ast=u_i(\alpha_1,\ldots,\alpha_{n+1})$ for appropriate positive Borel functions
$t_i(\cdot),\ u_i(\cdot),\ i=1,\ldots,n$. Succinctly, we write
$\mathbf{X}\backsim Ga_{1,\ldots,n}(\boldsymbol{\gamma}^\ast,\boldsymbol{\alpha}^\ast)$,
where
$\boldsymbol{\gamma}^\ast=(\gamma_1^\ast,\ldots,\gamma_n^\ast)'\in\mathbf{R}_+^n$
and
$\boldsymbol{\alpha}^\ast=(\alpha_1^\ast,\ldots,\alpha_n^\ast)'\in\mathbf{R}_+^n$
are vectors of parameters.
\end{definition}
%\begin{example}[Cheriyan (1941), Ramabhadran (1051)]
%Let $Y_k\backsim G_k(\gamma_k,\alpha),\ k=0,\ldots,n$ be mutually independent gamma
%random variables with common scale, and
%choose
%the matrix $A$ such that $(A)_{i,1}=A_{i,i+1}\equiv 1,$ for $i=1,\ldots,n$ and
% zero otherwise. Then
% $\mathbf{X}\backsim G_{1,\ldots,n}(\boldsymbol{\gamma}^\ast,\boldsymbol{\alpha}^\ast)$,
% where $\boldsymbol{\gamma}^\ast=(\gamma_{0}+\gamma_1,\ldots,\gamma_{0}+\gamma_n)'$
% and $\boldsymbol{\alpha}^\ast=(\alpha,\ldots,\alpha)'$ are two $n$-variate vectors of
% parameters; here $t_i=\gamma_{0}+\gamma_i$ and $u_i\equiv\alpha,\ i=1,\ldots,n$.
%\end{example}
\begin{example}[Mathai and Moschopoulos, 1991; see also Cherian, 1941 and Ramabhadran, 1951]
 \label{ExMMG1991}
Let $Y_j\backsim Ga(\gamma_j,\alpha_j),\ j=1,\ldots,n+1$ be mutually independent
random variables  distributed gamma  with arbitrary parameters, and
choose
the matrix $A$ such that, for
$i=1,\ldots,n$ and $\sigma_i>0$, it holds that $(A)_{i,n+1}=\alpha_{n+1}/\alpha_i$, $(A)_{i,i}\equiv 1$ and zero otherwise. Then
 $\mathbf{X}\backsim Ga_{1,\ldots,n}(\boldsymbol{\gamma}^\ast,\boldsymbol{\alpha})$,
 where $\boldsymbol{\gamma}^\ast=(\gamma_{n+1}+\gamma_1,\ldots,\gamma_{n+1}+\gamma_n)'$
 and $\boldsymbol{\alpha}=(\alpha_1,\ldots,\alpha_n)'$ are two $n$-variate vectors of
 parameters.
\end{example}
\begin{example}[Mathai and Moschopoulos, 1992; see  also, Furman, 2008]
\label{ExFG2009}
Let $Y_j\backsim Ga(\gamma_j,\alpha_j)$, $j=1,\ldots,n$ be mutually independent
random variables distributed gamma with arbitrary parameters, and
choose
the matrix $A$ such that, for
$i=1,\ldots,n,\ j=1,\ldots,i$ and $\sigma_i>0$, it holds that
 $(A)_{i,j}=\alpha_{j}/\sigma_i$  and zero otherwise.
Then
 $\mathbf{X}\backsim Ga_{1,\ldots,n}(\boldsymbol{\gamma}^\ast,\boldsymbol{\sigma})$,
 where $\boldsymbol{\gamma}^\ast=(\gamma^\ast_1,\ldots,\gamma^\ast_n)',\ \gamma^\ast_i=
 \sum_{j=1}^i \gamma_j$
 and $\boldsymbol{\sigma}=(\sigma_1,\ldots,\sigma_n)'$ are two $n$-variate vectors of
 parameters.
\end{example}

In the present paper we employ the following modification of Examples \ref{ExMMG1991} and
\ref{ExFG2009}.
\begin{example}[Furman, 2008]
\label{ExFG2009gen}
Let $Y_j\backsim Ga(\gamma_j,\alpha_j),\ j=1,\ldots,n+1$ be again mutually independent
random variables distributed gamma with arbitrary parameters, and
choose
the matrix $A$ such that, for
$i=1,\ldots,n,\ j=1,\ldots,i$ and $\sigma_i>0$, it holds that
 $(A)_{i,j}=\alpha_{j}/\sigma_i$, $(A)_{i,n+1}=\alpha_{n+1}/\sigma_i$  and zero otherwise.
 Then
 $\mathbf{X}\backsim Ga_{1,\ldots,n}(\boldsymbol{\gamma}^\ast,\boldsymbol{\sigma})$,
 where $\boldsymbol{\gamma}^\ast=(\gamma^\ast_1,\ldots,\gamma^\ast_n)',\ \gamma^\ast_i=
 \gamma_{n+1}+\sum_{j=1}^i \gamma_j$
 and $\boldsymbol{\sigma}=(\sigma_1,\ldots,\sigma_n)'$ are two $n$-variate vectors of
 parameters.
\end{example}

In the sequel, we embark on the idea in Example \ref{ExFG2009gen} to introduce
an encompassing yet tractable multivariate distribution with univariate
margins distributed Pareto.
We note in passing that a real-valued
r.v.
is said to be distributed Pareto of the $2$nd kind, succinctly
$X\backsim Pa(II)(\mu,\sigma,\alpha),$ where $\mu\in\mathbf{R}$ is a location parameter,
$\sigma\in\mathbf{R}_+$ is a scale parameter and $\alpha\in\mathbf{R}_+$ is a tail index,
 if its c.d.f. is given by
\begin{equation}
\label{Par2cdf}
F_X(x;\mu,\sigma,\alpha)=1-\left(
1+\frac{x-\mu}{\sigma}
\right)^{-\alpha},\ x> \mu
\end{equation}
(see, e.g., Pareto, 1897; Arnold, 1983; Kotz et al., 2000).
Similarly to Asimit et al. (2010),
we set $\mu=0$, which conveniently makes the support of the distribution
coincide with the positive half of the real line, i.e.,
$supp_F=\{x\in\mathbf{R}\ :f(x)\neq 0\}=\mathbf{R}_{+}$ and does not lead
to any loss of generality. The resulting distribution (Lomax distribution), notationally $Pa(II)(\sigma,\alpha)$, enjoys
a great variety of applications in all areas of applied mathematics in general and in
actuarial science in particular, as it naturally arises in the extreme value theory
as the limiting distribution of the excess-of-loss r.v. $X_d:=X-d|\ X>d$
where
$d(\in\mathbf{R}_+)$ denotes a threshold
(see, e.g., Balkema and de Haan, 1974; Pickands, 1975).

The rest of the paper is organized as follows. In Section \ref{sec2} a multivariate probability structure
with dependent Pareto-distributed  univariate margins is introduced and linked to a number of existing multivariate models.
Then distributional properties of the new structure are derived and some characterization results are proved in
Sections \ref{sec2} and \ref{sec-biv-stuff}.
In Section \ref{SecAppl} the new multivariate Pareto
is reintroduced as a variant of the minima-based multiple risk factor model, and some applications
to notions of actuarial interest are considered.
In Section \ref{sec-numex} an application of the model is elucidated with the help of a
numerical example borrowed from the context of default risk.
Section \ref{sec-con} concludes the paper. All proofs are relegated
to the appendix to facilitate the reading.

\section{New multivariate Pareto distribution }
\label{sec2}

Let $\mathbf{Y}=(Y_1,\ldots,Y_{n+1})'$ be a r.v. with mutually independent coordinates
$Y_j\backsim Ga(\gamma_j,1),$ $\gamma_j\in\mathbf{R}_+$, and choose
the matrix $A_c\in Mat_{n\times (n+1)}$ such
that $(A_c)_{i,j}=c_{i,j}/\sigma_i$, where
$c_{i,j}\in\{0,\ 1\}$ are deterministic constants, $\sigma_i\in\mathbf{R}_+$,
$i=1,\ldots,n$ and $j=1,\ldots,n+1$.
The following definition
unifies Examples \ref{ExMMG1991}, \ref{ExFG2009} and \ref{ExFG2009gen} and serves as an auxiliary tool for constructing the multivariate Pareto distribution of interest.
\begin{definition}
\label{GamDefG}
Let $\mathbf{X}=(X_1,\ldots,X_n)'=A_c\mathbf{Y}$, then it follows an $n$-variate
gamma distribution, notationally $\mathbf{X}\backsim Ga_{1,\ldots,n}(\boldsymbol{\gamma}_c^\ast,\
\boldsymbol{\sigma})$, where $\boldsymbol{\gamma}_c^\ast=(\gamma_{c,1}^\ast,\ldots,\gamma_{c,n}^\ast)'\in\mathbf{R}_+^n$ with $\gamma_{c,i}^\ast=\sum_{j=1}^{n+1} c_{i,j}\gamma_j,\ i=1,\ldots,n$ and
$\boldsymbol{\sigma}=(\sigma_1,\ldots,\sigma_n)'\in\mathbf{R}_+^n$ are two vectors of
parameters.
\end{definition}

We note in passing that Definition \ref{GamDefG} (auxiliary for the present paper)
establishes an encompassing multivariate probability law with gamma-distributed univariate
margins and an additive background risk dependence structure (see, Gollier and Pratt, 1996;
Tsanakas, 2008; Furman and Landsman, 2010 for applications of the additive background risk models in
economics and actuarial science).  More specifically,
the following simple special cases of
$Ga_{1,\ldots,n}(\boldsymbol{\gamma}_c^\ast,\
\boldsymbol{\sigma})$ readily recover
the models of, respectively, Mathai and Moschopoulos (1991,\ 1992) and Furman (2008):
\begin{itemize}
\item $c_{i,i}=c_{i,n+1}\equiv 1$ for $i=1,\ldots,n$ and zero otherwise - Example 1.1;
\item $c_{i,j}\equiv1$ for $1\leq j\leq i\leq n$ and zero otherwise - Example 1.2;
\item $c_{i,j}=c_{i,n+1}\equiv1$ for $1\leq j\leq i\leq n$
and zero otherwise - Example 1.3.
\end{itemize}

Some elementary but useful properties of $\mathbf{X}\backsim Ga_{1,\ldots,n}$
follow directly by definition or from the Laplace transform that is established next.
\begin{proposition}
\label{propLTG}
Let $\mathbf{X}\backsim Ga_{1,\ldots,n}(\boldsymbol{\gamma}_c^\ast,\
\boldsymbol{\sigma})$ be the r.v. distributed multivariate gamma as in Definition \ref{GamDefG},
then the corresponding Laplace transform is given by
\[
\hat{G}_{1,\ldots,n}(\mathbf{t})=\prod_{j=1}^{n+1}\left(
1+\sum_{i=1}^n \frac{c_{i,j}}{\sigma_i}t_i
\right)^{-\gamma_j},
\]
and it is well-defined  on $\mathbf{R}_{+}^n$.
\end{proposition}

Immediate consequences of Proposition \ref{propLTG} are, for $k,l=1,\ldots,n$, that
\begin{itemize}
\item the distribution of $\mathbf{X}\backsim Ga_{1,\ldots,n}(\boldsymbol{\gamma}_c^\ast,\
\boldsymbol{\sigma})$
is `marginally closed', i.e., $X_k\backsim Ga(\gamma_{c,k}^\ast(\in\mathbf{R}_+),\ \sigma_k(\in\mathbf{R}_+))$;
\item the expectation of the $k$-th coordinate is $\mathbf{E}[X_k]=\gamma_{c,k}^\ast/\sigma_k$;
\item the variance of the $k$-th coordinate is $\mathbf{Var}[X_k]=\gamma_{c,k}^\ast/\sigma_k^2$;
\item for $k\neq l$, the covariance between the coordinates
$X_k$ and $X_l$ is non-negative and
given by
\[
\mathbf{Cov}[X_k,X_l]=\frac{\sum_{j=1}^{n+1} c_{k,j}c_{l,j}\gamma_j}{\sigma_k\sigma_l};
\]
\item for $k\neq l$, the Pearson linear correlation between the coordinates
$X_k$ and $X_l$ is non-negative and given by
\[
\rho[X_k,X_l]=
\frac{\sum_{j=1}^{n+1} c_{k,j}c_{l,j}\gamma_j}{\sqrt{\gamma_{c,k}^\ast\gamma_{c,l}^\ast}}.
\]
\end{itemize}

We are now in a position to introduce the multivariate Pareto distribution of interest.
In fact, simple observation (\ref{GamLS}) along with Proposition \ref{propLTG} result in the following definition.
\begin{definition}
\label{defMPflexGen}
We call the r.v. $\mathbf{X}=(X_1,\ldots,X_n)'$ having the
decumulative distribution function
(d.d.f.)
\begin{equation}
\label{MPflexGenddf}
\overline{F}_{1,\ldots,n}(x_1,\ldots,x_n)=\prod_{j=1}^{n+1}\left(
1+\sum_{i=1}^n \frac{c_{i,j}}{\sigma_i}x_i
\right)^{-\gamma_j}, \textnormal{ where }(x_1,\ldots,x_n)'\in\mathbf{R}_{+}^n,
\end{equation}
a multivariate Pareto of the $2$nd kind. Succinctly, we write
$\mathbf{X}\backsim Pa_{1,\ldots,n}^c$
$(\boldsymbol{\sigma},\ \boldsymbol{\gamma},\ \gamma_{n+1})$, where $\boldsymbol{\sigma}=(\sigma_1,\ldots,\sigma_n)',\ \boldsymbol{\gamma}=(\gamma_1,\ldots,\gamma_n)'$ are two deterministic vectors of positive
parameters, and $\gamma_{n+1}\in\mathbf{R}_+$ and $c\in Mat_{n\times (n+1)}(\{0,\ 1\})$ are scalar-valued
and matrix-valued parameters, respectively.
\end{definition}

Generally,  distributions with Paretian tails have been applied in a multitude of areas. Herein we refer to:
 Benson et al. (2007) for applications in modelling catastrophic risk; Koedijk et al. (1990), Longin (1996), Gabaix et al. (2003) for applications in general financial phenomena; Cebri\'{a}n et al. (2003) for
applications in insurance pricing; and Soprano et al. (2010), Chavez-Demoulin et al. (2015)
for applications in risk management.

Specifically, the probability law in Definition \ref{defMPflexGen} is a generalization of the classical
multivariate Pareto distribution of Arnold (1983) with the d.d.f. $\overline{F}_{1,\ldots,n}^{Arnold}$. Indeed, set
 $c_{i,j}=c_{\bullet,j},\ i=1,\ldots,n,\ j=1,\ldots,n+1$ in (\ref{MPflexGenddf}) and
 obtain, for $\gamma^\ast_{c}=\sum_{j=1}^{n+1}c_{\bullet,j}\gamma_j$, that
\begin{equation}
\label{MPArnold}
(2.1)=
\left(
1+\sum_{i=1}^n \frac{x_i}{\sigma_i}
\right)^{-{\gamma}^\ast_c}=\overline{F}^{Arnold}_{1,\ldots,n}(x_1,\ldots,x_n), \textnormal{ where }(x_1,\ldots,x_n)'\in\mathbf{R}_{+}^n.
\end{equation}
That being said, unlike the classical multivariate Pareto distribution of Arnold (1983),
the structure in Definition \ref{defMPflexGen} incorporates stochastic independence -
set $c_{i,i}\equiv 1,\ i=1,\ldots, n$ and zero otherwise and obtain, for $\overline{F}_{1,\ldots,n}^\Pi$ denoting
the d.d.f. of a multivariate Pareto with independent margins, that
\[
(2.1)
 = \prod_{i=1}^{n}\left(
1+\frac{x_i}{\sigma_i}
\right)^{-\gamma_i}=\overline{F}^{\Pi}_{1,\ldots,n}(x_1,\ldots,x_n), \textnormal{ where }(x_1,\ldots,x_n)'\in\mathbf{R}_{+}^n.
\]
Consequently, the new multivariate Pareto distribution
meaningfully fills the gap between the multivariate probability distributions with independent
and Arnold-dependent Pareto-distributed margins.

In addition,  unlike (\ref{MPArnold}), $Pa_{1,\ldots,n}^c(\boldsymbol{\sigma},\ \boldsymbol{\gamma},\ \gamma_{n+1})$ allows for
distinct marginal tail indices (see, Proposition \ref{Nt1} below). Furthermore, the new multivariate Pareto distribution unifies the
 probability models
studied recently in Chiragiev and Landsman
 (2009). Namely, in order to obtain their `flexible Pareto type I and II' we choose $c_{i,i}=c_{i,n+1}\equiv 1,\ i=1,\ldots,n$ and zero
 otherwise and $c_{i,j}\equiv 1$ for $1\leq j\leq i\leq n$ and zero otherwise, respectively.

 Lastly but perhaps most importantly in actuarial applications, d.d.f. (\ref{MPflexGenddf}) admits stochastic representations
 that mimic the
 multiplicative background risk model (Franke et al., 2006) and the minima-based
 common shock model (Bowers et al., 1997) (see, respectively, Theorems \ref{CharLem} and \ref{MinimaCS}
 in this paper).
Stochastic representations are a very welcome facet, since they endow probabilistic models with an important feature of interpretability,
and as a result contribute greatly to the process of model selection and implementation.

We further document several simple properties of the multivariate Pareto with d.d.f.
(\ref{MPflexGenddf}).
The proofs are straightforward and thus omitted.
\begin{proposition}
\label{Nt1}
Let $\mathbf{X}\backsim Pa_{1,\ldots,n}^c$
$(\boldsymbol{\sigma},\ \boldsymbol{\gamma},\ \gamma_{n+1})$ as in Definition
\ref{defMPflexGen}, then, for $i=1,\ldots,n$, the marginal d.d.f. of $X_i$ is
\[
\overline{F}_i(x_i)=
\left(1+\frac{x_i}{\sigma_i}\right)^{-\gamma_{c,i}^\ast},\ x_i\in\mathbf{R}_+,
\]
that is $X_i\backsim Pa(II)(\sigma_i,\ \gamma_{c,i}^\ast)$, where
$\gamma_{c,i}^\ast=\sum_{j=1}^{n+1}c_{i,j}\gamma_j$. Also, for $i=1,\ldots,n$ and setting
 $\gamma_{c,i}^\ast>1$, we have that
\begin{equation}
\label{Epar}
\mathbf{E}[X_i]=\sigma_i/\left(\gamma_{c,i}^\ast-1\right),
\end{equation}
and furthermore setting $\gamma_{c,i}^\ast>2$,  we obtain that
\begin{equation}
\label{Varpar}
\mathbf{Var}[X_i]=\sigma_i^2\gamma_{c,i}^\ast/
\left(\left(
\gamma_{c,i}^\ast-1
\right)^2
\left(
\gamma_{c,i}^\ast-2
\right)
\right).
\end{equation}
\end{proposition}

In Proposition \ref{Nt1},
the substitution $c_{i,i}=c_{i,n+1}\equiv 1,\ i=1,\ldots,n$ and zero otherwise yields Theorem 1 of Chiragiev and Landsman
 (2009), whereas the substitution $c_{i,j}\equiv 1$ for $1\leq j\leq i\leq n$ and zero otherwise results in
 their Theorem 5.

In what follows, we develop an expression for the joint p.d.f. of the multivariate Pareto distribution of interest. To this end,
let
\begin{equation}\label{pdfaux}
\prod_{i=1}^n\sum_{j=1}^{n+1}c_{i,j} y_j =\sum_{i_j\in I} d_c(i_1,\ldots,i_{n+1})\prod_{j=1}^{n+1} y_j^{i_j},
\end{equation}
where $I$ establishes a set of positive integer indices such that $\sum_{j=1}^{n+1}i_j=n$, and $d_c(i_1,\ldots,i_n)$ are appropriately
chosen constants. Also, let
\[
(\gamma)_n=\frac{\Gamma(\gamma+n)}{\Gamma(\gamma)},\ \textnormal{ where }
\gamma\in\mathbf{R}_+ \textnormal{ and } n\in\mathbf{N}
\]
denote the Pochhammer symbol.
\begin{theorem}
\label{jointpdf}
Let $\mathbf{X}\backsim Pa_{1,\ldots,n}^c(\boldsymbol{\sigma},\ \boldsymbol{\gamma},\ \gamma_{n+1})$ as in Definition \ref{defMPflexGen},
then the corresponding joint p.d.f. is formulated, for $(x_1,\ldots,x_n)'\in \mathbf{R}_+^n$, as
 \begin{eqnarray}
 \label{joint_pdf}
 f_{1,\ldots,n}(x_1,\ldots,x_n)
  = \sum_{\forall i_j\in I}d_c(i_1,\ldots,i_{n+1})\prod_{j=1}^{n+1}\frac{(\gamma_j)_{i_j}}{\prod_{l=1}^n\sigma_l }\left(1+\sum_{i=1}^nc_{i,j}\frac{x_i}{\sigma_i}\right)^{-(\gamma_j+i_j)},
 \end{eqnarray}
 where $d_c(i_1,\ldots,i_{n+1})$ are appropriately chosen constants and $\ i_j\in I$.
\end{theorem}
In general, the constants $d_c(i_1,\ldots,i_n)$ can be rather involved.  For an insight,  we show
how (\ref{joint_pdf}) reduces to the p.d.f. of the classical multivariate Pareto distribution of Arnold (1983).
To this end,
set $c_{i,j}\equiv 1$ for $i=1,\ldots,n$ and $j=1\ldots,n+1$.
Then from (\ref{pdfaux}), we
have that
\[
d_c(i_1,\ldots,i_{n+1})={n \choose i_1,\ldots, i_{n+1}}
\]
with the right-hand side denoting the multinomial coefficient.
On the other hand, as (\ref{joint_pdf}) must integrate to one and since for the Arnold's
multivariate Pareto distribution, we have, for $\gamma^\ast=\gamma_1+\cdots+\gamma_{n+1}$,
that
\[
\prod_{j=1}^{n+1}\left(1+\sum_{i=1}^nc_{i,j}\frac{x_i}{\sigma_i}\right)^{-(\gamma_j+i_j)}
=\left(1+\sum_{i=1}^n\frac{x_i}{\sigma_i}\right)^{-(\gamma^\ast+n)},
\]
we obtain
\[
(2.6)=\frac{(\gamma^\ast)_n}{\prod_{i=1}^n \sigma_i}\left(1+\sum_{i=1}^n\frac{x_i}{\sigma_i}\right)^{-(\gamma^\ast+n)},\ \textnormal{for}
\ (x_1,\ldots,x_n)'\in\mathbf{R}_+^n,
\]
as required.

The following theorem establishes
a useful characteristic relation  in the context of the multivariate Pareto distribution of interest, and
it also plays an important role when deriving the formula for the corresponding Pearson linear
correlation (see, Theorem \ref{cov} in Section \ref{sec-biv-stuff}). In the sequel `$\overset{d}{=}$' denotes equality in
distribution.
\begin{theorem}
\label{CharLem}
Let $\boldsymbol{\Lambda}=(\Lambda_1,\ldots,\Lambda_n)'$ be a r.v. with independent
and exponentially-distributed
univariate margins $\Lambda_i\backsim Exp(1),\ i=1,\ldots,n$, and denote by
$\boldsymbol{\Xi}=(\Xi_1,\ldots,\Xi_n)'$ $\backsim Ga_{1,\ldots,n}(\boldsymbol{\gamma}_c^\ast,\
\boldsymbol{\sigma})$ the $n$-variate gamma distribution introduced in Definition
\ref{GamDefG}; here $\boldsymbol{\gamma}_c^\ast=(\gamma_{c,1}^\ast,\ldots,\gamma_{c,n}^\ast)'\in\mathbf{R}_+^n$ with $\gamma_{c,i}^\ast=\sum_{j=1}^{n+1} c_{i,j}\gamma_j$, and
$\boldsymbol{\sigma}=(\sigma_1,\ldots,\sigma_n)'\in\mathbf{R}_+^n$ are vectors of
parameters. Assume that $\boldsymbol{\Lambda}$ and $\boldsymbol{\Xi}$ are stochastically independent,
then $\mathbf{X}=(X_1,\ldots,X_n)'$ has d.d.f.
(\ref{MPflexGenddf}), and it is thus the $n$-variate Pareto distribution introduced in
Definition \ref{defMPflexGen} if and only if
$(X_1,\ldots,X_n)'\overset{d}{=}(\Lambda_1/ \Xi_1,\ldots,\Lambda_n/ \Xi_n)'$.
\end{theorem}

Theorem \ref{CharLem} establishes the multiplicative background risk representation
of the multivariate probabilistic structure of main interest herein
(see, Franke et al., 2006; Meyers, 2007, Asimit et al., 2013, 2016
for applications of the multiplicative
background risk models in economics and actuarial science).

We conclude this section with yet another characterization of the multivariate Pareto
distribution of interest and its two implications. Namely,
let $\wedge_{i=1}^n X_i=:X_{-}\backsim F_{-}$ and
$\vee_{i=1}^n X_i=:X_+\sim F_+$ denote, respectively the minima and the maxima r.v.'s,
and let $X_i\sim F_i,\ i=1,\ldots,n$ be univariate coordinates of the multivariate Pareto r.v.
of interest in this paper.

\begin{theorem}
\label{minima}
Let $\mathbf{X}=(X_1,\ldots,X_n)'$ be distributed $Pa_{1,\ldots,n}^c(\boldsymbol{\sigma},\ \boldsymbol{\gamma},\ \gamma_{n+1})$ as per Definition \ref{defMPflexGen}, then
$X_{-}$ admits the mixture representation as $X_-|\Lambda=\lambda\backsim Exp(\lambda)$ and
$\Lambda\overset{d}{=}Z_1+\cdots+Z_{n+1}$, where $Z_j,\ j=1,\ldots,n+1$ are univariate
mutually independent r.v.'s distributed gamma.
\end{theorem}

An important corollary of Theorem
\ref{minima} is a random parameter representation
(see, e.g., Feller, 1966)
of the d.d.f.'s of $X_-$ and $X_+$.
 The following lemma is crucial in studying the distribution of $X_-$ in Theorem \ref{ggamma}.

\begin{lemma}[Moschopoulos, 1985; Furman and Landsman, 2005]
\label{FL2005} For $i=1,\ldots,n$, let $Z_i\sim Ga(\gamma_i(\in\mathbf{R}_+).\ \alpha_i(\in\mathbf{R}_+))$ denote independent gamma-distributed
 r.v.'s. Then the distribution of $Z=Z_1+\cdots+Z_n$ is gamma with a random shape
parameter. More specifically, $Z\sim Ga(\gamma^\ast+K,\ \alpha_{+}),$ where
$\gamma^\ast=\gamma_1+\cdots+\gamma_n$, $\alpha_{+}=\vee_{i=1}^n \alpha_i$ and $K$ is an
integer-valued non-negative r.v. with the probability mass function (p.m.f.) given by
\begin{equation}
\label{pk}
p_k=\mathbf{P}[K=k]=c_{+} \delta_k,\ k=0,\ 1, \ldots,
\end{equation}
where
\[
c_{+}=\prod_{i=1}^n \left(\frac{\alpha_i}{\alpha_{+}}\right)^{\gamma_i},\
\delta_0=1
\]
and
\[
\delta_k=k^{-1}\sum_{l=1}^k\sum_{i=1}^n \gamma_i \left(
1-\frac{\alpha_i}{\alpha_{+}}
\right)^l\delta_{k-l} \textnormal{ for } k>0.
\]
\end{lemma}

\begin{theorem}
\label{ggamma} Let $\mathbf{X}\backsim Pa_{1,\ldots,n}^c(\boldsymbol{\sigma},\ \boldsymbol{\gamma},\ \gamma_{n+1})$ as in Definition \ref{defMPflexGen},
then
$X_{-}\backsim Pa(II)(\alpha_{+}(\boldsymbol{\sigma}),\\ \gamma^\ast +K)$, where
$\alpha_{+}(\boldsymbol{\sigma})=\vee_{j=1}^{n+1} \left(\sum_{i=1}^n\frac{c_{i,j}}{\sigma_i}\right)^{-1}$, $K$ is an integer-valued
r.v. with p.m.f. (\ref{pk})
%\begin{eqnarray*}
%q_k=\frac{\sigma_{+} p_k}{\mathbf{E}[\Lambda^{-1}](\gamma^\ast+k-1)},
%\end{eqnarray*}
and $\gamma^\ast=\gamma_1+\cdots+\gamma_{n+1} >1$.
\end{theorem}

While Theorem \ref{ggamma} demonstrates that the minima r.v. $X_-$ is distributed mixed Pareto with random
tail index parameter, the next theorem asserts that the maxima r.v. $X_+$ has a d.d.f.
that is a linear combination of the d.d.f.'s of such mixed Pareto-distributed r.v.'s.
The proof is similar to the one of Proposition 2 in Vernic (2011) and is thus omitted.
\begin{theorem}
\label{maxrv}
Assume that $\mathbf{X}\backsim Pa_{1,\ldots,n}^c(\boldsymbol{\sigma},\ \boldsymbol{\gamma},\ \gamma_{n+1})$ as in Definition \ref{defMPflexGen}, then
the d.d.f. of the maxima r.v.
is given by
\begin{equation}
\label{maxddf}
\overline{F}_+(x)=\sum_{\mathcal{S}\subseteq\{1,\ldots,n\}} (-1)^{|\mathcal{S}|-1}
\overline{F}_{\mathcal{S}-}(x),\ x\in\mathbf{R}_{ +},
\end{equation}
where $X_{\mathcal{S}-}=\wedge_{s\in\mathcal{S}\subseteq\{1,\ldots,n\}}X_s$ and
$X_{\mathcal{S}-}\sim F_{\mathcal{S}-}$.
\end{theorem}

\section{Bivariate quantities of interest}
\label{sec-biv-stuff}

It is worthwhile to make an additional observation before stating the main result of
this section. Namely, we note in passing that for $1\leq k\neq l\leq n$,
a r.v. $(\Xi_k,\ \Xi_l)'$ distributed the bivariate gamma per Definition
\ref{GamDefG}
and an $(n+1)$-variate r.v.
$\mathbf{Y}=(Y_1,\ldots,Y_{n+1})'$ having mutually independent coordinates
$Y_j\backsim Ga(\gamma_j,1),$ $\gamma_j\in\mathbf{R}_+$,
the
following stochastic representation holds
\begin{equation}
\label{StRep}
(\sigma_k \Xi_k,\ \sigma_l \Xi_l)'
\overset{d}{=}
(Y_{c,(k,l)}+Y_{c,k},\ Y_{c,(k,l)}+Y_{c,l})',
\end{equation}
where $Y_{c,(k,l)}=\sum_{j=1}^{n+1}c_{k,j}c_{l,j}Y_j$, $Y_{c,k}=\sum_{j=1}^{n+1}c_{k,j}(1-c_{l,j})Y_j$ and $Y_{c,l}=\sum_{j=1}^{n+1}c_{l,j}(1-c_{k,j})Y_j$ are mutually independent
gamma-distributed r.v.'s with the shape parameters $\gamma_{c,(k,l)}=\sum_{j=1}^{n+1}c_{k,j}c_{l,j}\gamma_j$,
$\gamma_{c,k}=\sum_{j=1}^{n+1}c_{k,j}(1-c_{l,j})\gamma_j$ and
$\gamma_{c,l}=\sum_{j=1}^{n+1}c_{l,j}(1-c_{k,j})\gamma_j$, respectively.  %{\color{blue}The the utmost general form of the bivariate d.d.f. in Definition \ref{defMPflexGen} is written as, for $x_k,x_l\in \mathbf{R}_+$
%\[
%\overline{F}_{k,l}(x_k,x_l)=\left(1+\frac{x_k}{\sigma_k}\right)^{-\gamma_{c,k}}
%\left(1+\frac{x_l}{\sigma_l}\right)^{-\gamma_{c,l}}\left(1+\frac{x_k}{\sigma_k}+\frac{x_l}{\sigma_l}\right)^{-\gamma_
%{c,(k,l)}},
%\]
%and its probability density function is written as
%\begin{eqnarray*}
%&&f_{k,l}(x_k,x_l)=\frac{\partial^2}{\partial x_k \partial x_l}\overline{F}_{k,l}(x_k,x_l)\\
%&&=\overline{F}_{k,l}(x_k,x_l)\left[\frac{\gamma_{c,k}}{\sigma_k}\left(1+\frac{x_k}{\sigma_k}\right)^{-1}
%+\frac{\gamma_{c,l}}{\sigma_l}\left(1+\frac{x_l}{\sigma_l}\right)^{-1}+\gamma_{c,(k,l)}\left(\frac{1}{\sigma_k}+\frac
%{1}{\sigma_l}\right)\left(1+\frac{x_k}{\sigma_k}+\frac{x_l}{\sigma_l}\right)^{-1}\right].
%\end{eqnarray*}
%}

We next show that the covariance of a random couple within the multivariate Pareto of
interest in this paper can be formulated using the $(q+1)\times q$ hypergeometric function
(see, Gradshteyn and Ryzhik, 2007), which is formulated as
\begin{eqnarray}
\label{3F2}
_{q+1}F_q(a_1,\ldots,a_{q+1};b_1,\ldots,b_q;z):=\sum_{k=0}^{\infty}\frac{(a_1)_k,\ldots,(a_{q+1})_k }{(b_1)_k,\ldots,(b_q)_k}\frac{z^k}{k!},\
\end{eqnarray}
where $q\in\mathbf{Z}_+$.
For $a_1,\ldots,a_{q+1}$ all positive, and these are the cases of interest in the present paper,
the radius of convergence of the series is the open disk $|z|<1$. On the boundary $|z|=1$, the series
converges absolutely if $h:=b_1+\cdots + b_q-a_1-\cdots -a_{q+1}>0$, and it
converges except at $z=1$ if $0\geq h>-1$.
\begin{theorem}
\label{cov}
Let $\mathbf{X}\backsim Pa_{1,\ldots,n}^c(\boldsymbol{\sigma},\ \boldsymbol{\gamma},\ \gamma_{n+1})$ as in Definition \ref{defMPflexGen} and assume that
both $\gamma_{c,k}^\ast$ and $\gamma_{c,l}^\ast$ exceed two, then, for
$0\leq k\neq l\leq n$,
\[
\mathbf{Cov}[X_k,\ X_l]=
\sigma_k\sigma_l
\frac{1}{(\gamma_{c,k}^\ast-1)(\gamma_{c,l}^\ast-1)}
\left({}_3F_2\left(\gamma_{c,(k,l)},1,1;\gamma_{c,k}^\ast,\gamma_{c,l}^\ast;1\right)-1\right).
\]
\end{theorem}

An immediate consequence of Theorem \ref{cov} is that the maximal attainable Pearson
correlation in the context of the multivariate Pareto distribution introduced in the
present paper is not one. This consequence is however solely a
result of the fact that the
Pearson index of correlation exists only if the involved second moments are finite, a
pitfall that is well-known to non-life
actuaries, which often deal with heavy-tailed losses (see, Embrechts et al., 2002).

\begin{corollary}
\label{corrbounds}
Let $\mathbf{X}\backsim Pa_{1,\ldots,n}^c(\boldsymbol{\sigma},\ \boldsymbol{\gamma},\ \gamma_{n+1})$ and assume that
both $\gamma_{c,k}^\ast$ and $\gamma_{c,l}^\ast$ exceed two for
$0\leq k\neq l\leq n$, then, for the Pearson correlation, it holds that $\mathbf{Corr}[X_k,\
X_l]\in [0,\ 1/2)$.
\end{corollary}

Another consequence of Theorem \ref{cov} pertains to two special cases of the
multivariate Pareto introduced in this paper, and it is formulated as the following corollary.
\begin{corollary}
\label{CL2009Cov}
Let $\mathbf{X}_1=(X_{1,1},\ldots, X_{1,n})'\sim Pa_{1,\ldots,n}^{(I)}$ and
$\mathbf{X}_2=(X_{2,1},\ldots,X_{2,n})'\sim Pa_{1,\ldots,n}^{(II)}$
be distributed, respectively, the multivariate flexible Pareto of type $I$ and $II$ of Chiragiev and
Landsman (2009). Then the corresponding covariances are readily obtained,
for $0\leq k \neq l\leq n$, as
\begin{eqnarray}
\label{covflex1}
&&\mathbf{Cov}[X_{1,k},\ X_{1,l}]  \\
&=&\frac{\sigma_k\sigma_l}{(\gamma_k+\gamma_{n+1}-1)(\gamma_l+\gamma_{n+1}-1)}
\left({}_3F_2\left(\gamma_{n+1},1,1;\gamma_k+\gamma_{n+1},\gamma_l+\gamma_{n+1};1\right)-1\right) \notag,
\end{eqnarray}
for $\gamma_k+\gamma_{n+1}>2$ and $\gamma_l+\gamma_{n+1}>2$,
and as
\begin{eqnarray}
\label{covflex2}
\mathbf{Cov}[X_{2,k},\ X_{2,l}]&=& \frac{\sigma_k\sigma_l}{(\gamma_{c,k}^\ast-1)(\gamma_{c,l}^\ast-1)(\gamma_{c,l}^\ast-2)},
%&=&\frac{1}{(\gamma_{c,l}^\ast-1)(\gamma_{c,k}^\ast-1)}
%\left({}_3F_2\left(\gamma_{c,n+1},1,1;\gamma_{c,n+1},\gamma_{c,k}^\ast;1\right)-1\right)\nonumber \\
%&=&\frac{1}{(\gamma_{c,l}^\ast-1)(\gamma_{c,k}^\ast-1)}
%\left({}_2F_1\left(1,1;\gamma_{c,k}^\ast;1\right)-1\right)=???
\end{eqnarray}
for $\gamma_{c,k}^\ast=\sum_{j=1}^k\gamma_j>2$ and $\gamma_{c,l}^\ast=\sum_{j=1}^l\gamma_j>2$.
\end{corollary}

 We note in passing that expression (\ref{covflex2}) confirms the one derived
in Chiragiev and Landsman (2009), whereas formula (\ref{covflex1}) complements the discussion
therein. Also, the covariance of two r.v.'s coming from the Arnold's multivariate
Pareto distribution (see, Arnold, 1983) is readily obtained from both ($\ref{covflex1}$) and
($\ref{covflex2}$). More specifically, we set $\gamma_k=\gamma_l\equiv 0$
for all $0\leq k \neq l\leq n$ and verify that
(\ref{covflex1}) reduces to
\[
\frac{\sigma_k\sigma_l}{(\gamma_{n+1}-1)(\gamma_{n+1}-1)}
\left({}_2F_1\left(1,1;\gamma_{n+1};1\right)-1\right)=
\frac{\sigma_k\sigma_l}{(\gamma_{n+1}-1)^2(\gamma_{n+1}-2)},
\]
for $\gamma_{n+1}>2$. The verification is straightforward in the case of (\ref{covflex2}).

Generalized hypergeometric function (\ref{3F2}) plays an important role when deriving the
centred regression function $r(y)=\mathbf{E}[X-\mathbf{E}[X]|\ Y=y]$, where $y\in\mathbf{R}_+$
(see, Furman and Zitikis, 2008b, 2010, for applications of the function in insurance and
finance).  We next present the conditional d.d.f., followed by the centered regression function for a pair of r.v.'s having the probability law
as in Definition \ref{defMPflexGen}. To this end, let
\[
m(x)=\frac{\sigma_k}{\gamma_{c,(k,l)}}\left(
1+\frac{x}{\sigma_l}
\right),\ x\in\mathbf{R}_+.
\]
\begin{theorem}
\label{cd}
Let $\mathbf{X}\backsim Pa_{1,\ldots,n}^c(\boldsymbol{\sigma},\ \boldsymbol{\gamma},\ \gamma_{n+1})$ as in Definition \ref{defMPflexGen}, the d.d.f. of $X_k$ given $X_l=x_l,\ 0\leq k\neq l\leq n$, is formulated as
\begin{eqnarray}
\label{CondFunction}
&&\mathbf{P}[X_k>x_k|X_l=x_l]\\
=&&\left(\frac{\gamma_{c,(k,l)}}{\gamma_{c,l}^\ast}+\frac{\gamma_{c,l}}{\gamma_{c,l}^\ast} \left(1+\frac{x_k}{\gamma_{c,(k,l)}m(x_l)} \right)\right)\left(1+\frac{x_k}{\sigma_k} \right)^{-\gamma_{c,k}}\left(1+\frac{x_k}{\gamma_{c,(k,l)}m(x_l)} \right)^{-\gamma_{c,(l,k)}-1}, \nonumber
\end{eqnarray}
where $x_k,x_l\in \mathbf{R}_+$.
\end{theorem}
%The proof of Proposition \ref{cd} is immediate by definition, and it is thus omitted.

\begin{theorem}
\label{cregTh}
Let $\mathbf{X}\backsim Pa_{1,\ldots,n}^c(\boldsymbol{\sigma},\ \boldsymbol{\gamma},\ \gamma_{n+1})$ as in Definition \ref{defMPflexGen}, the centred regression function of $X_k$ on $X_l,\ 0\leq k\neq l\leq n$, is given, for $\gamma_{c,k}^\ast>1$, by
\begin{equation}
\label{rk}
r_k(x_l)=m(x_l)\sum_{i=1}^2 a_i\ _2F_1\left({\gamma}_{c,k},1;{\gamma}_{c,k}^\ast+2-i;-\frac{x_l}{\sigma_l}  \right)-\sigma_k/(\gamma_{c,k}^\ast -1),
\end{equation}
where
\[
a_1=\frac{{{\gamma}^2_{c,(k,l)}}}{{\gamma}^*_{c,k}~ {\gamma}^*_{c,l}},\
a_2=\frac{{\gamma}_{c,l} ~{\gamma}_{c,(k,l)}}{{\gamma}^*_{c,l} ({\gamma}^*_{c,k}-1)}\
 \textnormal{ and } x_l\in\mathbf{R}_{+}.
\]
The centred regression function is monotonically-increasing and concave.
\end{theorem}
{\color{blue} We reiterate that our results readily recover the ones derived in Landsman and Chiragiev (2009). More specifically, 
by a simple alignment of notation in Theorem \ref{cregTh} above, we obtain Theorem 3 in loc. cit., whereas by choosing   
$\gamma_{c,k}=0$ in Theorem \ref{cregTh} and hence for
\[
{}_2F_1\left(0,1;{\gamma}_{c,k}^\ast+2-i;-\frac{x_l}{\sigma_l}  \right)=1,\ 
a_1=\frac{\gamma_{c,k}^\ast}{\gamma_{c,l}^\ast} \textnormal{ and }
a_2=\frac{(\gamma_{c,l}^\ast-\gamma_{c,k}^\ast)\gamma_{c,k}^\ast}{\gamma_{c,l}^\ast(\gamma_{c,k}^\ast-1)},
\]
we end up with Theorem 7 therein.}

Clearly, the centred regression function of the new multivariate Pareto distribution
 is not linear, while it is well-known that the classical multivariate Pareto has linear regression
(see, Arnold, 1983).  Theorem \ref{cregTh} confirms the latter fact by
setting $\gamma_{c,(k,l)}=\gamma_{c,k}^\ast=\gamma_{c,l}^\ast\equiv \gamma_{n+1}$ and $\gamma_{c,l}=\gamma_{c,k} \equiv 0$ in (\ref{rk}), which then reduces to the following
linear form
\[
r_k(x_l)=\frac{\sigma_k}{\sigma_l}\left(x_l-\frac{\sigma_l}{\gamma_{n+1}(\gamma_{n+1}-1)}\right)
\]
for $\gamma_{n+1}>1,\ 0\leq k\neq l\leq n$ and $x_l\in\mathbf{R}_+$.

\section{Applications to insurance}
\label{SecAppl}
In what follows, we assume that $\mathbf{X}=(X_1,\ldots,X_n)'$ denotes a risk portfolio
(r.p.) with $X_i,\ i=1,\ldots,n$ representing its risk components (r.c.'s). According to Theorem
\ref{CharLem}, if $\mathbf{X}\sim
 Pa_{1,\ldots,n}^c(\boldsymbol{\sigma},\ \boldsymbol{\gamma},\ \gamma_{n+1})$, then
it admits the multiplicative background risk representation (see, Franke et al., 2006;
 Meyers, 2007; Asimit et al., 2013, 2016).

We next show that the new multivariate Pareto distribution can also be interpreted
as a variant of the classical
minima-based common shock model (see, e.g., Bowers et al,
1997). To this end,
assume that the $i$-th r.c of the r.p. is exposed to the set $\mathcal{R}_i=\{r\in\mathbf{N}:r\leq (n+1)\}$,
 $i=1,\ldots,n$ of risk factors (r.f.'s) and
let the r.v.
$\mathbf{Y}=(Y_1,\ldots,Y_{n+1})'$ stipulate the randomness of actuarial interest associated
with the  r.f.'s.
The following theorem
establishes the minima-based
multiple risk factor representation of the multivariate Pareto proposed in the present paper.
We note in
passing that `$*$'  stands for the mixture operator, i.e., given two appropriately
jointly measurable r.v.'s
${X}_{{\lambda}}\sim C(\cdot;{\lambda})$ and ${\Lambda}\sim H$,
it holds that ${X}_{{\lambda}}*{\Lambda}\overset{d}{=}{X}_{{\Lambda}}$.

\begin{theorem}
\label{MinimaCS}
Let $\mathbf{W}_i=(W_{i,1},\ldots,W_{i.n+1})'$ be r.v.'s with
independent
exponentially-distributed margins
$W_{i,j}\sim Exp(\lambda_{i,j}(\in\mathbf{R}_+)),\ i=1,\ldots,n,\ j=1,\ldots,n+1$, and
let $A_c$ be a deterministic matrix
of zero-one coefficients. Also, let $\mathbf{\Lambda}=(\Lambda_1,\ldots,\Lambda_{n+1})'$
be a r.v. having independent gamma-distributed margins with arbitrary shape parameters
$\gamma_j(\in\mathbf{R}_+)$ and rate parameters equal to $1$,
$j=1,\ldots,n+1$. Set, for $\sigma_i \in \mathbf{R}_+$ and  $i=1,\ldots,n$,
\begin{equation}\label{minX-eq}
X_i=\sigma_i
\bigwedge_{j=1,\ c_{i,j}\neq 0}^{n+1}
(W_{i,j}* \Lambda_j),
\end{equation}
then
$\mathbf{X}=(X_1,\ldots,X_n)'\sim Pa_{1,\ldots,n}^c(\boldsymbol{\sigma},\ \boldsymbol{\gamma},\ \gamma_{n+1})$.
\end{theorem}

Theorem \ref{MinimaCS} suggests that the multivariate Pareto distribution
proposed in the present paper might be an appropriate formal framework for modelling
dependent default, survival or failure times when these times are exponentially-distributed
with random parameters. We elaborate on this observation in Section
\ref{sec-numex}.

\subsection{Actuarial risk measurement}
Regulatory accords around the globe require that insurance companies carry out a careful
assessment of their future losses. From now on, the r.v. $X:\Omega \rightarrow \mathbf{R}_+$ is
interpreted as an insurance loss r.v., and $\mathcal{X}$ denotes the collection of such r.v.'s.
\begin{definition}
\label{DefRM}
A risk measure is a functional map
$H:\mathcal{X}\rightarrow[0,\ \infty]$.
\end{definition}

The literature on risk measures is
vast and growing quickly. The following two
indices are arguably the most popular amongst practitioners.
\begin{definition}
\label{VaRDef}
Let $X\in\mathcal{X}$ and fix $q\in[0,\ 1)$, then the Value-at-Risk (VaR) and the Conditional
Tail Expectation (CTE) risk measures are respectively given by
\begin{equation}
\label{VaRdef}
VaR_q[X]=\inf\{x\in\mathbf{R}:\mathbf{P}[X\leq x]\geq q\}
\end{equation}
and
\begin{equation}
\label{CTEdef}
CTE_q[X]=\mathbf{E}[X|\ X>VaR_q[X]].
\end{equation}
\end{definition}
\noindent
We note in passing that both VaR and CTE are distorted as well as weighted risk measures (see, respectively, Wang, 1996; and Furman and
Zitikis, 2008a).
\begin{definition}{Furman and Zitikis (2008a), see also Choo and de Jong (2009, 2010).}
\label{Weightedpcp}
Let $w:\mathbf{R}\rightarrow \mathbf{R}_+$ be a
non-decreasing Borel (weight) function, such
that $0<\mathbf{E}[w(X)]<\infty$,
then the class of weighted risk measures is defined as
\begin{equation}
\label{wpcp}
\pi_{w}[X]=\frac{\mathbf{E}[Xw(X)]}{\mathbf{E}[w(X)]}
\textnormal{ for }X\in\mathcal{X}.
\end{equation}
\end{definition}
Let $v_1,\ w_1:\mathbf{R}^n\rightarrow \mathbf{R}_+$ be two
legitimate weight functions such that all expectations in (\ref{vwpcp}) are finite,
and consider a generalized variant of (\ref{wpcp})
 \begin{equation}
\label{vwpcp}
\pi_{v_1,\ w_1}[\boldsymbol{\Lambda}]=
\frac{\mathbf{E}[v_1(\boldsymbol{\Lambda})w_1(\boldsymbol{\Lambda})]}{\mathbf{E}[w_1(\boldsymbol{\Lambda})]}\
\textnormal{ for }\boldsymbol{\Lambda}:\Omega\rightarrow \mathcal{R}\subseteq \mathbf{R}^n.
\end{equation}
\begin{proposition}
\label{CTEmixt}
Let $X|\boldsymbol{\Lambda} =\boldsymbol{\lambda} \backsim C(\cdot;
\boldsymbol{\lambda}(\in\mathcal{R}\subseteq\mathbf{R}^n))$ and assume that $\boldsymbol{\Lambda} \backsim H_{1,\ldots,n}$, then for any legitimate
weight function, the functional $\pi_w[X]$ admits representation (\ref{vwpcp}).
\end{proposition}

\begin{corollary}
\label{VaRProp}
Let $\mathbf{X}\backsim Pa_{1,\ldots,n}^c(\boldsymbol{\sigma},\ \boldsymbol{\gamma},\ \gamma_{n+1})$ as in Definition \ref{defMPflexGen}, then, for $i=1,\ldots,n$ and $q\in[0,\ 1)$,
 we have that
\[
VaR_q[X_i]=\sigma_i\left(
(1-q)^{-1/\gamma_{c,i}^\ast}-1
\right).
\]
\end{corollary}

\begin{corollary}
\label{CTEgen}
Let $X|\boldsymbol{\Lambda} =\boldsymbol{\lambda} \backsim C(\cdot;\
\boldsymbol{\lambda}(\in\mathcal{R}\subseteq\mathbf{R}^n))$ and assume that $\boldsymbol{\Lambda} \backsim H_{1,\ldots,n}$, then
the CTE risk measure of $X\backsim F$ is, if exists and for $q\in[0,\ 1)$, given by
\[
CTE_q[X]=\frac{\mathbf{E}[\overline{C}(VaR_q[X];\boldsymbol{\Lambda})CTE_{q^\ast}[X|\boldsymbol{\Lambda}]]}{\overline{F}(VaR_q[X])}  \textnormal{ for } X\in\mathcal{X},
\]
where $q^\ast=C(VaR_q[X];\ \boldsymbol{\lambda})$.
\end{corollary}

\begin{corollary}
\label{CTEParXiCor}
Let $\mathbf{X}\backsim Pa_{1,\ldots,n}^c(\boldsymbol{\sigma},\ \boldsymbol{\gamma},\ \gamma_{n+1})$ as in Definition \ref{defMPflexGen}, we have that the CTE risk measure is,
if exists and  for $i=1,\ldots,n$ and $q\in[0,\ 1)$, given by
\begin{eqnarray}
\label{CTEpar}
CTE_q[X_i]&=&\mathbf{E}[X_i] \frac{\overline{F}_{X_i^\ast}(VaR_q[X_i])}{1-q}+VaR_q[X_i] \nonumber\\
&=&\mathbf{E}[X_i]+VaR_q[X_i]\frac{\gamma_{c,i}^\ast}{\gamma_{c,i}^\ast-1},
\end{eqnarray}
where $X_i^\ast \sim Pa(II)(\sigma_i,\gamma_{c,i}^\ast-1)$.
\end{corollary}

%\begin{figure}
%  % Requires \usepackage{graphicx}
%  \includegraphics[width=10cm]{CTE_VaR_plots2.eps}\\
%  \caption{Plots with the VaR and the CTE risk Measures}\label{cteplot}
%\end{figure}

The minima r.v. $X_{-} = \wedge_{i=1}^n X_i$ plays an important role in insurance mathematics
(recall, e.g., the joint life policies in life insurance), as well as in general finance
(think of, e.g., the first-to-default baskets).

Recall that the r.v. $K$ has been defined as an integer-valued non-negative r.v.
with the following p.m.f.
 \begin{equation}
\label{pk2}
p_k=\mathbf{P}[K=k]=c_{+} \delta_k,\ k=0,\ 1, \ldots,
\end{equation}
where
\[
c_{+}=\prod_{i=1}^n \left(\frac{\alpha_i}{\alpha_{+}}\right)^{\gamma_i},\
\delta_0=1
\]
and
\[
\delta_k=k^{-1}\sum_{l=1}^k\sum_{i=1}^n \gamma_i \left(
1-\frac{\alpha_i}{\alpha_{+}}
\right)^l\delta_{k-l} \textnormal{ for } k>0.
\]
\begin{proposition}
\label{CTEparMin}
In the context of the multivariate Pareto of interest, the CTE risk measure of the
minima can be written, if finite and for $q\in[0,\ 1)$, as
\[
CTE_q[X_-]=\mathbf{E}[X_-]\frac{\overline{F}_{X_-^\ast}(VaR_q[X_{-}])}{1-q}+VaR_q[X_{-}],
\]
where $X_-^\ast \sim Pa(\alpha_+(\boldsymbol{\sigma}),\ \gamma^\ast+Q-1)$, where
$\alpha_+(\boldsymbol{\sigma})=\vee_{j=1}^{n+1}\left(
\sum_{i=1}^n \frac{c_{i,j}}{\sigma_i}
\right)^{-1}$, $\gamma^\ast=\gamma_1+\cdots+\gamma_{n+1}$ and $Q$ is an integer-valued r.v. with the p.m.f. obtained from the p.m.f.
of $K$ with the help of the following change of measure
\begin{equation}
\label{cmkq}
q_k=\frac{1}{\mathbf{E}[X_{-}]}\frac{\alpha_+(\boldsymbol{\sigma})}{\gamma^\ast+k-1}p_k,\
k=0,1,\ldots
\end{equation}
\end{proposition}
\begin{proposition}
\label{ctemaxima}
In the context of the multivariate Pareto of interest, the CTE risk measure of the
maxima can be written,  if finite and for $q\in[0,\ 1)$, as
the following linear combination
\[
CTE_q[X_+]=\frac{1}{1-q}\sum_{\mathcal{S}\subseteq \{1,\ldots,n\}} (-1)^{|\mathcal{S}|-1} \mathbf{E}[X_{\mathcal{S}-}|X_{\mathcal{S}-}>VaR_q(X_+)] \overline{F}_{{\mathcal{S}-}}(VaR_q(X_+)),
\]
where $X_{\mathcal{S}-}=\wedge_{s\in \mathcal{S}\subseteq\{1,\ldots,n\}}X_s$ and
$X_{\mathcal{S}-}\sim F_{\mathcal{S}-}$.
\end{proposition}

\begin{definition}[Furman and Zitikis, 2008b]
Let $w:\mathbf{R} \rightarrow \mathbf{R}_+$ be a non-decreasing Borel function, such that
$0<\mathbf{E}[w(Y)]<\infty$, then the functional $\Pi:\mathcal{X}\times\mathcal{X}
\rightarrow [0,\ \infty]$ is referred to as the economic risk measure. Moreover,
the special form of $\Pi$, given by
\begin{equation}
\label{ewpcp}
\Pi_w[X,\ Y]=\frac{\mathbf{E}[Xw(Y)]}{\mathbf{E}[w(Y)]}\
\textnormal{for\ }X\in\mathcal{X} \textnormal{ and }Y\in\mathcal{X},
\end{equation}
is called a weighted economic risk measure.
\end{definition}
We further derive an expression for the economic CTE risk measure, which is a particular
case of (\ref{ewpcp}) with $w(y)=\mathbf{1}\{y>VaR_q[Y]\},\ q\in[0,\ 1)$ and $y\in\mathbf{R}_+$.  To this
end, we find the next proposition useful. The proof is plain and thus omitted.
\begin{proposition}
\label{conditional_distribution}
Let $\mathbf{X}\backsim Pa_{1,\ldots,n}^c(\boldsymbol{\sigma},\ \boldsymbol{\gamma},\ \gamma_{n+1})$ as in Definition \ref{defMPflexGen}, the d.d.f.
 of $X_k$ given $X_l>x_l,\ 0\leq k\neq l\leq n$, is formulated as
\begin{eqnarray}
\label{CondFunction}
\mathbf{P}[X_k>x_k|X_l>x_l]=\left(1+\frac{x_k}{\sigma_k} \right)^{-\gamma_{c,k}}\left(1+\frac{x_k}{\gamma_{c,(k,l)}m(x_l)} \right)^{-\gamma_{c,(l,k)}}, \nonumber
\end{eqnarray}
where
\[
m(x_l)=\frac{\sigma_k}{{\gamma}_{c,(k,l)}}\left(1+\frac{x_l}{\sigma_l}\right)
\]
and
$x_k,\ x_l$ are both in $\mathbf{R}_+$.
\end{proposition}

\begin{proposition}
\label{econCTE}
In the context of the multivariate Pareto of interest,
the economic CTE risk measure is given, for $\gamma^\ast_{c,k}>1$,
 $q\in[0,\ 1)$ and $1\leq k \neq l\leq n$, by
 \[
\mathbf{E}[X_k|X_l > VaR_q[X_l]]=\frac{\sigma_k}{\gamma_{c,k}^*-1} {}_2F_1\left(\gamma_{c,(k,l)},1;\gamma_{c,k}^*;\frac{VaR_q[X_l]}{\sigma_l+VaR_q[X_l]}\right).
\]
\end{proposition}

To summarize, so far we have introduced and studied a new multivariate probability distribution with the univariate margins
distributed Pareto of the 2nd kind. The dependence structure of the new distribution is driven by a number
of stochastic representations that are variants of the multiplicative background risk and the minima-based
common shock models. We next employ the latter interpretation of the proposed multivariate probability
distribution to exemplify its possible application to modelling and measuring default risk.

\section{Numerical illustration}
\label{sec-numex}

{
For the sake of the discussion in this section,
we adopt the view of the
Financial Stability Board and the International Monetary Fund
that the systemic risk can be caused by impairment of
all or parts of the financial system, and more formally, we call the risk factor
$j\in\{1,\ldots,n+1\}$ `systemic', if $c_{i,j}=1$ for at least two distinct r.c.'s
$i\in\{1,\ldots,n\}$. Similarly, we call the risk factor $j\in\{1,\ldots,n+1\}$
`idiosyncratic', if $c_{i,j}=1$ for only one risk component $i\in\{1,\ldots,n\}$.

Consider obligors in a default risk portfolio,
each of
which is exposed to exactly two distinct categories of fatal risk factors, e.g.,
systemic (category A)
and idiosyncratic (category B).
We assume that the
risk factors from distinct risk categories are independent and that
the hitting times (or occurrences) of defaults of the r.c.'s are exponentially-distributed with random parameters distributed gamma.
In fact,
 the future lifetime r.v. of the $i$-th r.c. has exponential
distribution with the random parameter $\sigma_i^{-1}\sum_{j=1}^{n+1}c_{i,j}\Lambda_j$, where $\Lambda_j$ are
distributed gamma with unit rate parameters, and $i$ is $1,\ 2$ or 3.
Then Theorem \ref{MinimaCS} readily implies that the joint default times of the aforementioned
r.c.'s has d.d.f. (\ref{MPflexGenddf}).

To illustrate the effect of the dependence structure on the joint default
probability we further set the dimension to $n=3$ and
specialize the set-up above along the lines in Section 16.8 of Engelmann and Rauhmeier (2011)
as well as employing the 2014's Annual Global Corporate Default Study and Rating Transitions of
Standard \& Poor's. (Standard \& Poor's, 2015).
More specifically, we set
 $\mu:=\mathbf{E}[\Lambda_j]\equiv 1.67$, fix the time horizon to $15$
years and choose the corresponding default probability, $p$ say, to be equal to $0.3198$
(on par with the `B' credit rating of speculative entities).  This yields the multivariate
probability structure of Definition \ref{defMPflexGen} with identically distributed
margins having the parameters $\sigma_i\equiv \sigma=122.39$ and $\gamma_{c,i}^*\equiv 3.33$, for $i=1,\ 2$ and
$3$.

Then we explore three different exposures
of the obligors to the systemic
 and idiosyncratic  r.f.'s.
The distinct exposures
are stipulated by appropriate choices of the $c$ parameters gathered by matrices $A_c^{(k)},\
k=1,\ 2,\ 3$. We compare the aforementioned three exposures with the reference case
in which no systemic risk presents, that is the joint d.d.f. of default times
is a trivariate Pareto with independent margins.
 We note in passing that the expressions for the d.d.f.'s below readily follow
from Theorem \ref{MinimaCS}, whereas the values of the Pearson correlation coefficient
 are in non-trivial cases obtained  with the help of Theorem \ref{cov}.

\begin{itemize}
\item [Case (1).]
Only the systemic risk presents, and all risk components are exposed to it.  The exposure
is represented schematically with the use of the following matrix, in which the rows
and the columns represent r.c.'s and r.f.'s, respectively
\[
A_c^{(1)}=\left(
\begin{array}{cc|ccc}
1 & 1 & 0 & 0\\

1 & 1 & 0 & 0\\

1 & 1 & 0 & 0
\end{array}
\right).
\]
The joint d.d.f. of the risk components is given by
\[
\overline{F}^{(1)}(x_1,x_2,x_3)=\left(
1+\frac{x_1+x_2+x_3}{\sigma}
\right)^{-2\mu},
\]
where $x_1,x_2,x_3$ are all in $\mathbf{R}_{+}$.
This is obviously the d.d.f. of the classical trivariate Pareto distribution (Arnold, 1983).
In this r.p., the Pearson correlation coefficient between any two of the r.c.'s is $0.3$.
\end{itemize}
\noindent
In the following two cases, both the systemic and idiosyncratic risks present.
\begin{itemize}
\item [Case (2).]
There are overall three uncorrelated idiosyncratic risk factors and one
systemic risk factor.  The exposure is gathered by the following block matrix
\[
A_c^{(2)}=\left(
\begin{array}{c|cccc}
1 & 1 & 0 & 0\\

1 & 0 & 1 & 0\\

1 & 0 & 0 & 1
\end{array}
\right).
\]
The joint d.d.f. of the risk components is given by
\begin{eqnarray*}
&&\overline{F}^{(2)}(x_1,x_2,x_3)\\
&&=\left(
1+\frac{x_1+x_2+x_3}{\sigma}
\right)^{-\mu}
\left(
1+\frac{x_1}{\sigma}
\right)^{-\mu}
\left(
1+\frac{x_2}{\sigma}
\right)^{-\mu}
\left(
1+\frac{x_3}{\sigma}
\right)^{-\mu},
\end{eqnarray*}
where $x_1,x_2,x_3$ are all in $\mathbf{R}_{+}$. This case corresponds to the `flexible Pareto type I' of Landsman and
Chiragiev (2009).
In this r.p., the Pearson correlation coefficient between any two of the r.c.'s is $0.09$.

\item [Case (3).] The systemic
risk is represented by two distinct risk factors of which one targets the entire risk
portfolio and the other only hits r.c.'s $\# 1$ and $\# 2$. There is one idiosyncratic
risk factor, and only r.c. $\# 3$ is exposed to it.
The exposure block matrix is given by
\[
A_c^{(3)}=\left(
\begin{array}{cc|ccc}
1 & 1 & 0 & 0\\

1 & 1 & 0 & 0\\

1 & 0 & 1 & 0
\end{array}
\right).
\]
The joint d.d.f. of the risk components is
\[
\overline{F}^{(3)}(x_1,x_2,x_3)=\left(
1+\frac{x_1+x_2+x_3}{\sigma}\right)^{-\mu}
\left(
1+\frac{x_1+x_2}{\sigma}\right)^{-\mu}
\left(
1+\frac{x_3}{\sigma}
\right)^{-\mu},
\]
where $x_1,x_2,x_3$ are all in $\mathbf{R}_{+}$.
In this r.p., the Pearson correlation coefficient between r.c. \#1 and \#2  is $0.3$, and
it is equal to $0.09$ otherwise.
\end{itemize}

\subsection{Expected times of the first default} The left panel of
Figure \ref{fig:CTEmin}
depicts the values of $CTE_q[X_-]$ for $q\in[0,\ 1)$, $X_-\in\mathcal{X}$
and portfolios (1) to (3) as well as the reference portfolio, denoted by $(\perp)$.
{As the risk components are identically distributed, it is not difficult to see that
the following ordering holds
\begin{equation}
\label{stord1}
\overline{F}^{(1)}_-\geq_{st}\overline{F}^{(3)}_-\geq_{st}\overline{F}^{(2)}_-\geq_{st}
\overline{F}^{(\perp)}_-,
\end{equation}
where `$\geq_{st}$' denotes first order stochastic dominance (FSD). Furthermore,
since the CTE risk measure
is known to preserve the FSD ordering, we also have that
\[
CTE^{(1)}_q[X_-]\geq CTE^{(3)}_q[X_-]\geq CTE^{(2)}_q[X_-]\geq CTE^{(\perp)}_q[X_-]
\]
for all $q\in[0,\ 1)$ and $X_-\in\mathcal{X}$. This conforms with
Figure \ref{fig:CTEmin} (left panel), which hints that the r.p.'s with more
significantly correlated r.c.'s
enjoy higher, and thus more favourable, occurrence times of the first default.}
\begin{figure}[h!]
\centering
\includegraphics[width=7cm, height=7cm]{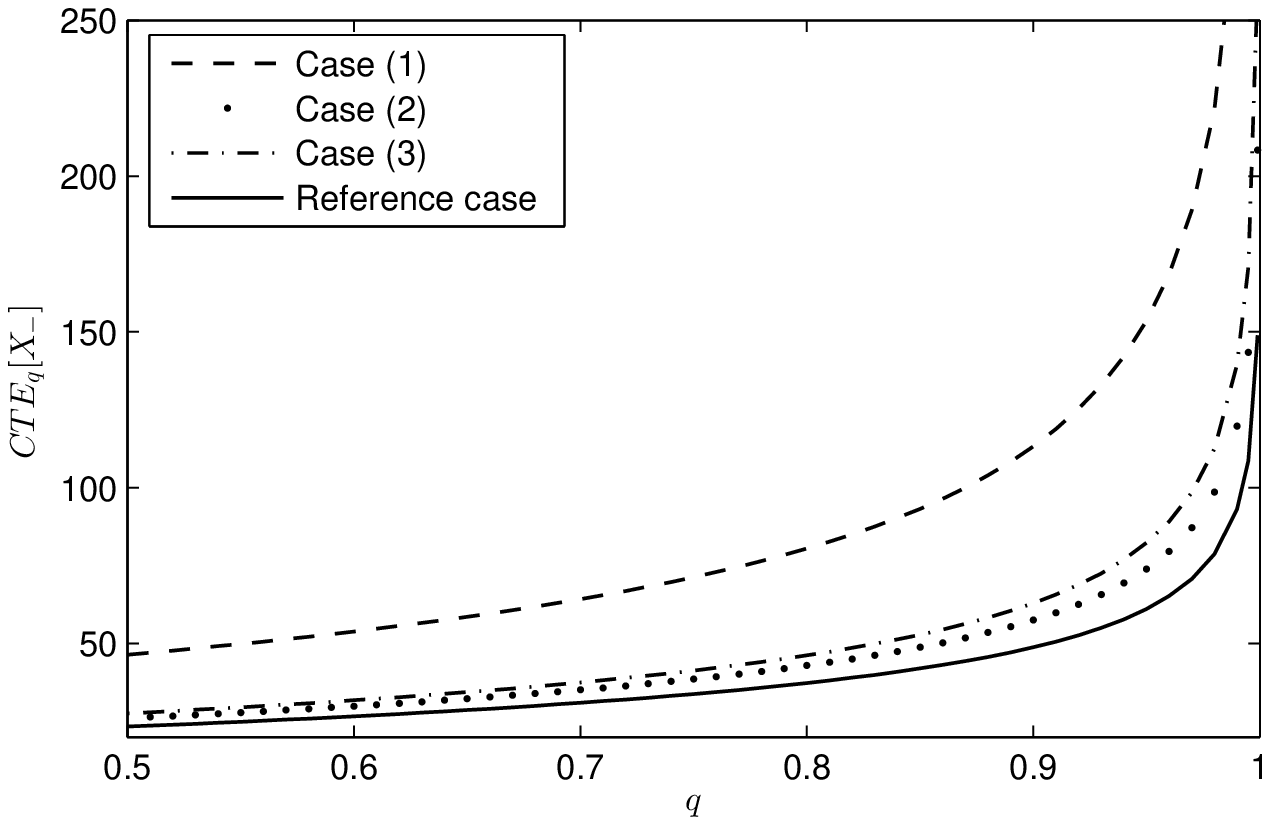}
 \includegraphics[width=7cm, height=7cm]{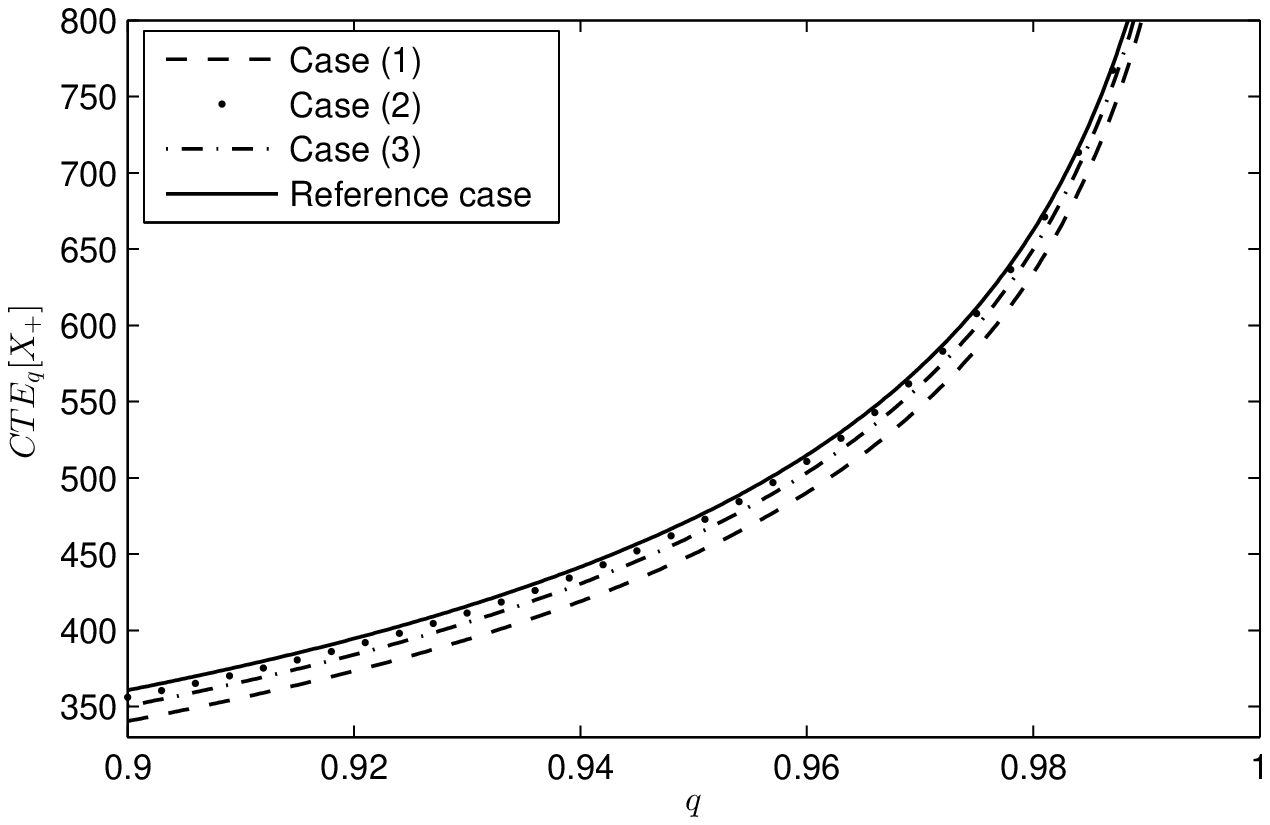}\\
\caption{Conditional expected times of first (left panel) and last
(right panel) default for portfolios (1) - (3) and the
reference portfolio $(\perp)$ for `B' rating r.p.'s with the probability of
default $p=0.3198$ and
$\mu=1.67$.
Proposition \ref{CTEparMin} is employed to
 compute the values of $CTE_q$ for $q\in[0,\ 1)$.}
\label{fig:CTEmin}
\end{figure}

The
downside of high correlations is elucidated in
 Figure \ref{fig:minCase}, in which
we leave the probability of default $p$ to be equal to $0.3198$ (`B' rating),
but vary the $\mu$ parameter that stipulates the effect of the risk factors.
In this respect, we observe that the r.p.'s with stronger correlations between r.c.'s are more
sensitive to the changes in the $\mu$ parameter, and therefore such r.p.'s
must be monitored and stress-tested more frequently.

\begin{figure}[h!]
\centering
\includegraphics[width=7cm,height=7cm]{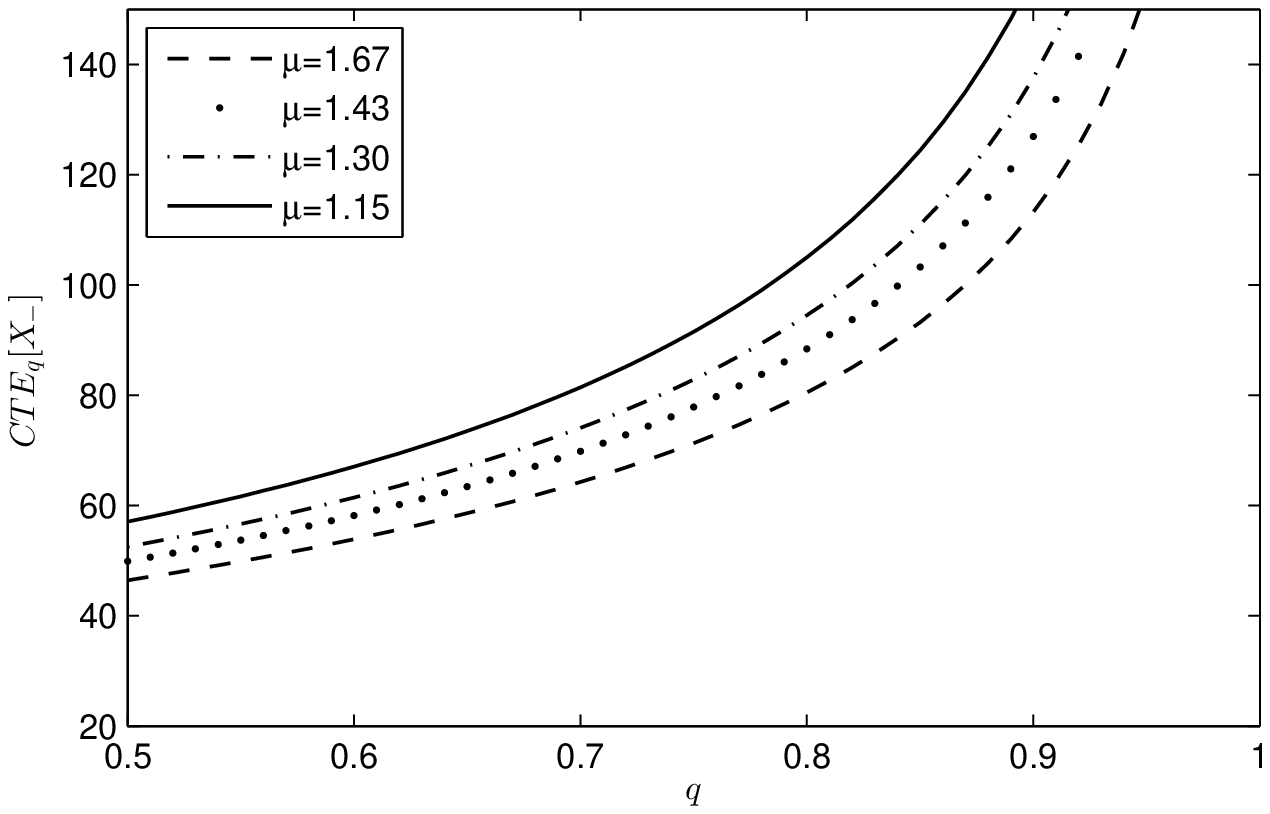}
\includegraphics[width=7cm,height=7cm]{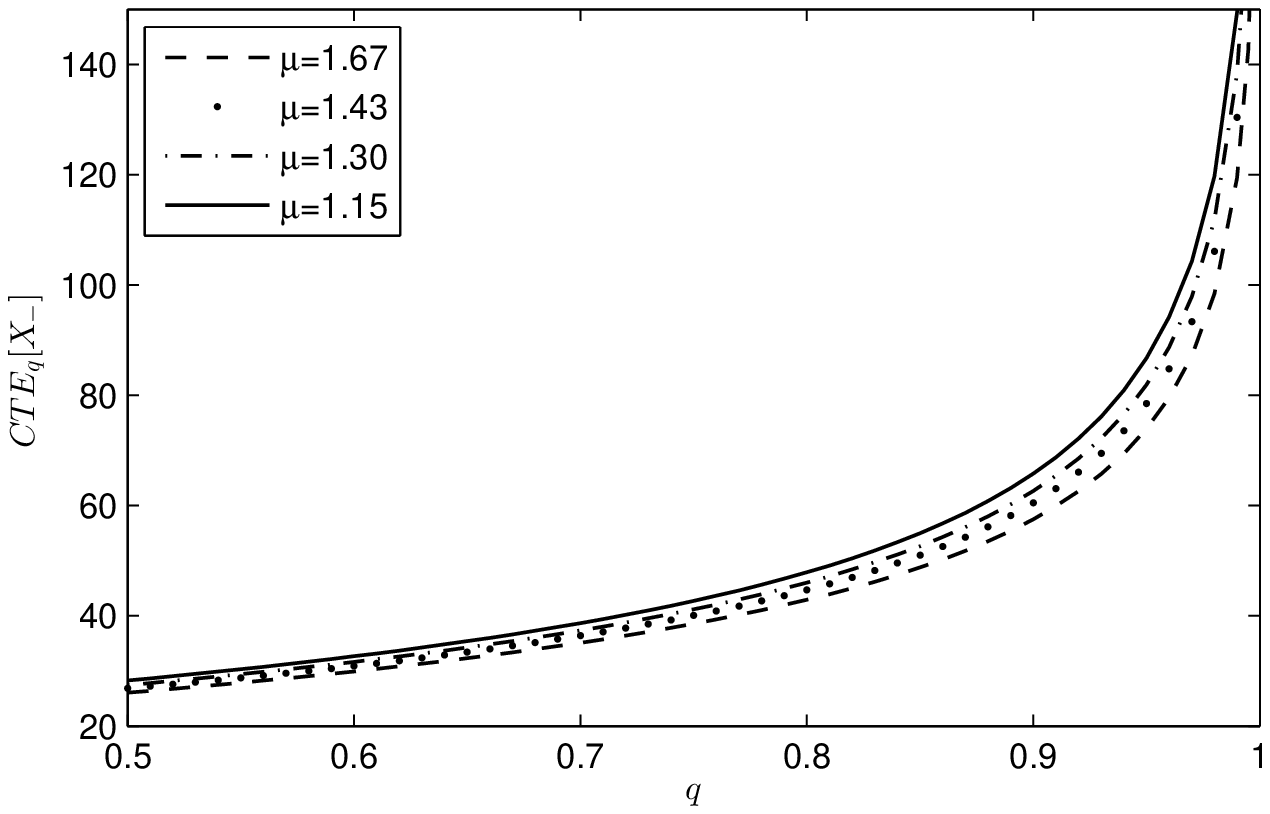}\\
\includegraphics[width=7cm,height=7cm]{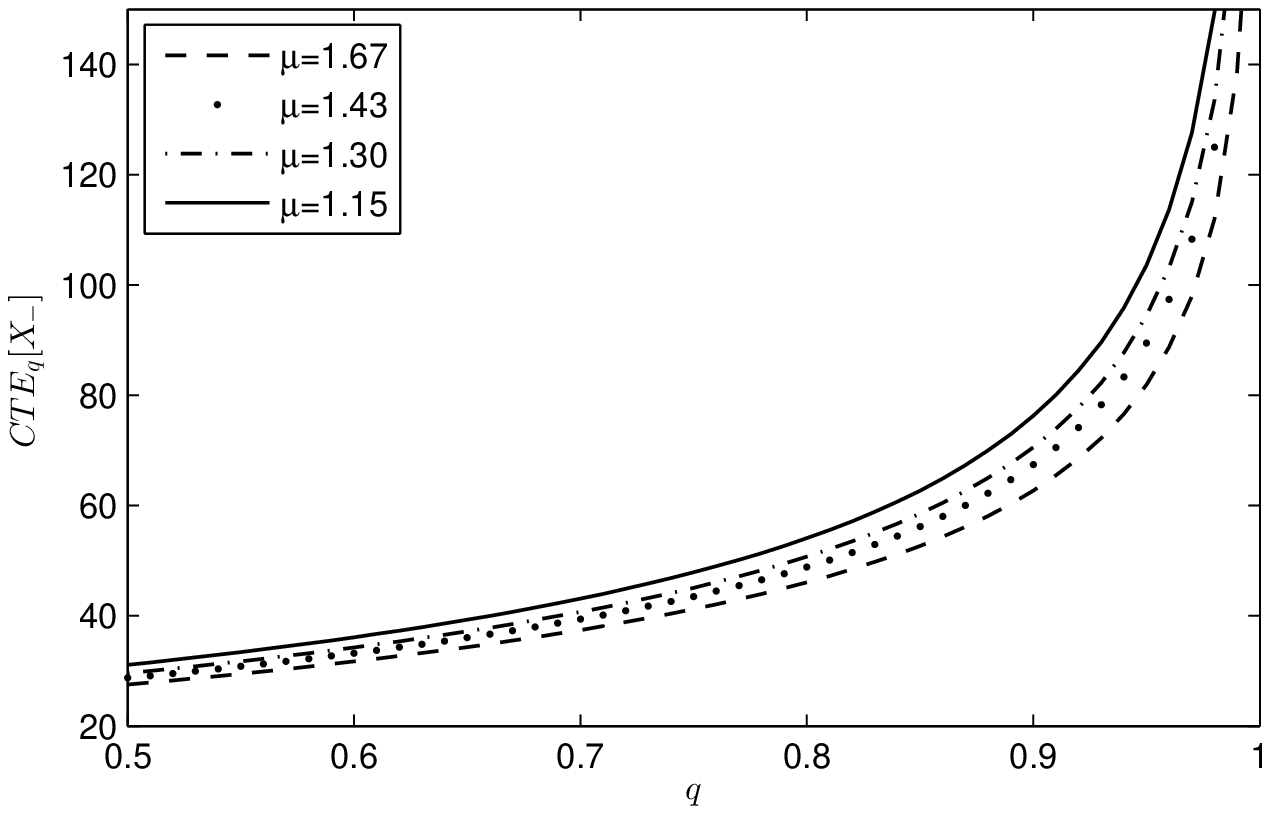}
\includegraphics[width=7cm,height=7cm]{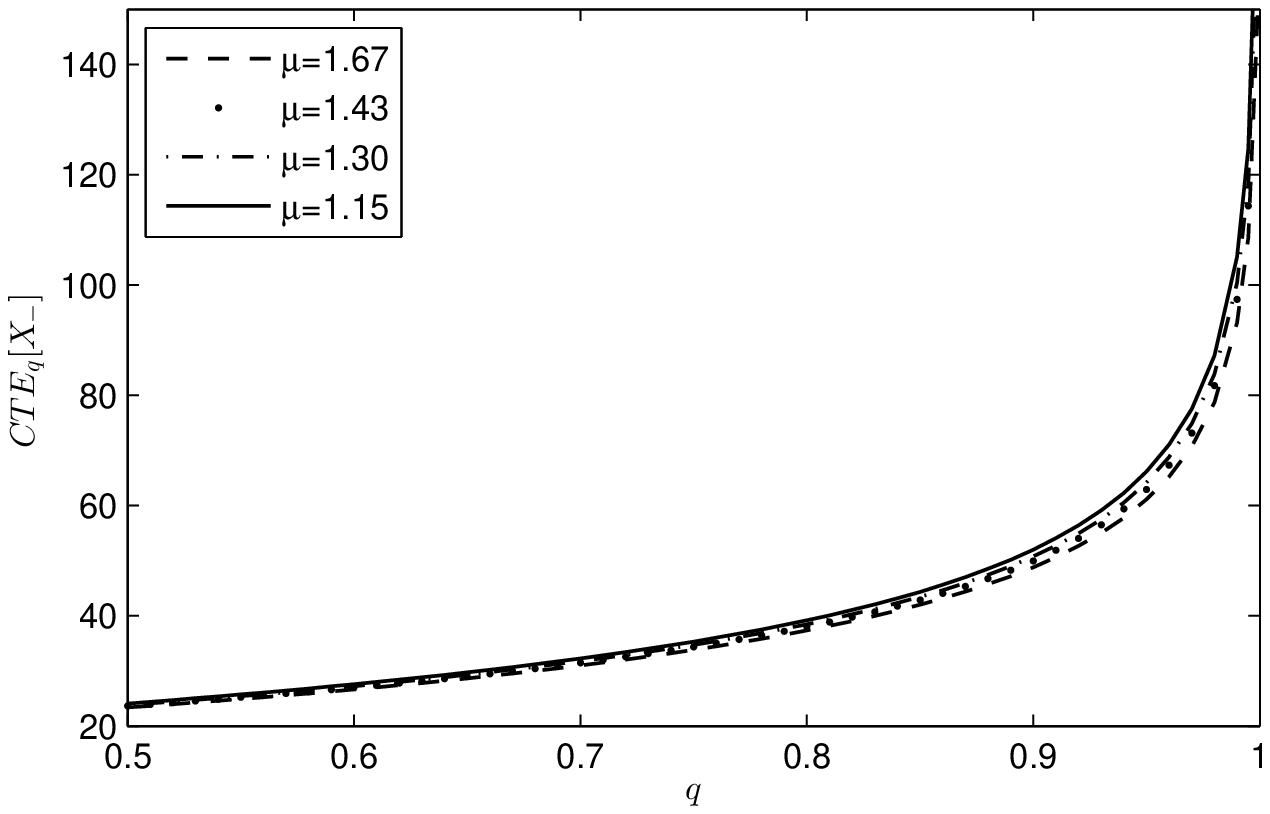}
\caption{Conditional expected times of first default for portfolios (1) (top left panel),
(2) (top right panel), (3) (bottom left panel) and reference $(\perp)$ (bottom right panel)
with the parameter $\mu$ varying from $1.67$ to $1.15$ and the default probability
$p=0.3198$. Proposition \ref{CTEparMin} is employed to compute the values of
$CTE_q$ for $q\in[0,\ 1)$.}
\label{fig:minCase}
\end{figure}

\subsection{Expected times of the last default}
Figure \ref{fig:CTEmin} (right panel)
depicts the values of $CTE_q[X_+]$ for $q\in[0,\ 1),$ $X_+\in\mathcal{X}$ and portfolios
(1) to (3) as well as the reference portfolio $(\perp)$.
Evoking Theorem \ref{maxrv} along with  (\ref{stord1}) results in
\[
\overline{F}^{(\perp)}_+\geq_{st}\overline{F}^{(2)}_+\geq_{st} \overline{F}^{(3)}_+
\geq_{st} \overline{F}^{(1)}_+,
\]
and hence
\[
CTE^{(\perp)}_q[X_+]\geq CTE^{(2)}_q[X_+]\geq CTE^{(3)}_q[X_+] \geq CTE^{(1)}_q[X_+]
\]
for all $q\in[0,\ 1)$ and $X_+\in\mathcal{X}$.
This conforms with the right panel of Figure \ref{fig:CTEmin}.

Unlike in the case of the first default,
we observe that if the time of the last default is
of interest and  the distributions of the r.c.'s are fixed,
then assuming stronger correlations between r.c.'s yields a more
conservative assessment of the expected time of the last default.
}
\section{Conclusions}
\label{sec-con}
We have introduced and studied a new form of an absolutely continuous with respect to the
Lebesgue measure multivariate probability law with the univariate margins distributed
Pareto of the 2nd kind.
The genesis of our distribution
is threefold, i.e., it originates as the Laplace transform of a multivariate gamma distribution
with the dependence structure based on the additive form of the multivariate reduction
method, and it also admits variants of the multiplicative background risk model
as well as the minima-based common shock model.  We have
meaningfully positioned the proposed multivariate Pareto distribution in the general context of
the current state of the art. We have proved and employed certain characteristic results to derive,
e.g., the (conditional/product) moments of the new multivariate Pareto as well as the
distributions of minima and maxima. Last but not least, we have developed
expressions for some tail-based risk measures of actuarial interest and elucidated our
findings with the help of a numerical example.

\section*{Acknowledgements}
We thank the anonymous referees and the Associate Editor, Prof. Montserrat Guill\' en,  for valuable comments
and suggestions that improved the work significantly and resulted in a better
presentation of the material.
We are also grateful to Prof. Paul Embrechts and all participants of the ETHs
Series of Talks in Financial and Insurance Mathematics for feedback and insights.

Our research has been supported by the Natural Sciences and Engineering Research Council (NSERC) of Canada. Jianxi Su also acknowledges the financial support of the Government of Ontario and MITACS Canada via, respectively, the Ontario Graduate Scholarship program
and the Elevate Postdoctoral fellowship.

\section*{References}
\hangindent=\parindent\noindent \textsc{Arnold, B.C. (1983)} \textit{Pareto
Distributions.} International Cooperative Publishing House, Fairland.

\hangindent=\parindent\noindent \textsc{Asimit, V., Furman, E. and Vernic, R. (2010)}
On a multivariate Pareto distribution.
\textit{Insurance: Mathematics and Economics} \textbf{46}(2), 308 -- 316.

\hangindent=\parindent\noindent \textsc{Asimit, V., Vernic, R. and Zitikis, R. (2013)}
Evaluating risk measures and capital allocations based on multi-losses driven by a heavy-tailed background risk: The multivariate Pareto-II model.
\textit{Risk} \textbf{1}(1), 14 -- 33.

\hangindent=\parindent\noindent \textsc{Asimit, V., Vernic, R. and Zitikis, R. (2016)}
Background risk models and stepwise portfolio construction.
\textit{Methodology and Computing in Applied Probability}, in press.

\hangindent=\parindent\noindent \textsc{Balkema, A. and de Haan, L. (1974)}
Residual life time at great age. \textit{Annals of Probability} \textbf{2}(5),
792 -- 804.

\hangindent=\parindent\noindent \textsc{Benson, D.A., Schumer, R. and Meerschaert, M.M. (2007)}
Recurrence of extreme events with power-law interarrival times.
\textit{Geophysical Research Letters} \textbf{34}(16), 1 -- 5.

\hangindent=\parindent\noindent \textsc{Boucher, J.P., Denuit, M. and Guill\'
en, M. (2008)} Models of insurance claim counts with time dependence based
on generalization of Poisson and negative binomial distributions. \textit{%
Variance} \textbf{2}(1), 135 -- 162.

\hangindent=\parindent\noindent \textsc{Bowers, N.L., Gerber, H.U., Hickman, J.C.,
Jones, D.A. and Nesbitt, C.J. (1997)} \textit{Actuarial Mathematics.} Second edition.
Society of Actuaries, Schaumburg.

\hangindent=\parindent\noindent \textsc{Cebri\'{a}n, A.C., Denuit, M. and Lambert, P. (2003)}
Generalized Pareto fit to the Society of Actuaries large claims database.
\textit{North American Actuarial Journal}
\textbf{7}(3), 18 -- 36.

\hangindent=\parindent\noindent \textsc{Chavez-Demoulin, V., Embrechts, P. and Hofert, M. (2015)}
An extreme value approach for modeling operational risk losses depending on covariates.
\textit{Journal of Risk and Insurance}, in press.

\hangindent=\parindent\noindent \textsc{Cherian, K.C. (1941)}
A bivariate correlated gamma-type distribution function.
\textit{Journal of the Indian Mathematical Society}
\textbf{5}, 133 -- 144.

\hangindent=\parindent\noindent \textsc{Chiragiev, A. and Landsman, Z. (2009)}
Multivariate flexible Pareto model: Dependency structure, properties and
characterizations.
\textit{Statistics and Probability Letters} \textbf{79}(16), 1733 -- 1743.

\hangindent=\parindent\noindent \textsc{Choo, W. and de Jong, P. (2009)}
Loss reserving using loss aversion functions.
\textit{Insurance: Mathematics and Economics} \textbf{45}(2), 271 -- 277.

\hangindent=\parindent\noindent \textsc{Choo, W. and de Jong, P. (2010)}
Determining and allocating diversification benefits for a portfolio of risks.
\textit{ASTIN Bulletin} \textbf{40}(1), 257 -- 269.

%\hangindent=\parindent\noindent \textsc{Dixon, A.C. (1902)}
%Summation of a certain series.
%\textit{Proc. London Math. Soc.} \textbf{35}(1), 284 -- 291.

\hangindent=\parindent\noindent \textsc{Embrechts, P., McNeil, A. and Straumann, D. (2002)}
Correlation and dependence in risk management: Properties and pitfalls. In: Dempster, M.
et al. (Eds), \textit{Risk Management: Value at Risk and Beyond}. Cambridge University Press,
Cambridge.

\hangindent=\parindent\noindent \textsc{Engelmann, B. and Rauhmeier, R. (2011)}
\textit{The Basel II Risk Parameters: Estimation, Validation, Stress Testing - with Applications to Loan Risk Management}. Springer, Berlin.

\hangindent=\parindent\noindent \textsc{Feller, W. (1966)}
\textit{An Introduction to Probability Theory and its Applications.}
John Wiley and Sons, New York.

\hangindent=\parindent\noindent \textsc{Franke, G., Schlesinger, H.
and Stapleton, R.C. (2006)}
Multiplicative background risk.
\textit{Management Science}
\textbf{52}(1), 146 -- 153.

\hangindent=\parindent\noindent \textsc{Furman, E. (2008)} On a multivariate
gamma distribution. \textit{Statistics and Probability Letters} \textbf{78}(15), 2353 -- 2360.

\hangindent=\parindent\noindent \textsc{Furman, E. and Landsman, Z. (2005)}
Risk capital decomposition for a multivariate dependent gamma portfolio.
\textit{Insurance: Mathematics and Economics} \textbf{37}(3), 635 -- 649.

\hangindent=\parindent\noindent \textsc{Furman, E. and Landsman, Z. (2010)}
Multivariate Tweedie distributions and some related capital-at-risk
analysis. \textit{Insurance: Mathematics and Economics}
\textbf{46}(2), 351 -- 361.

\hangindent=\parindent\noindent \textsc{Furman, E. and Zitikis, R. (2008a)}
Weighted premium calculation principles. \textit{Insurance: Mathematics and
Economics} \textbf{42}(1), 459 -- 465.

\hangindent=\parindent\noindent \textsc{Furman, E. and Zitikis, R. (2008b)}
Weighted risk capital allocations. \textit{Insurance: Mathematics and
Economics} \textbf{43}(2), 263 -- 269.

\hangindent=\parindent\noindent \textsc{Furman, E. and Zitikis, R. (2010)}
General Stein-type covariance decompositions with applications to
insurance and finance. \textit{ASTIN Bulletin} \textbf{40}(1), 369 -- 375.

\hangindent=\parindent\noindent \textsc{Gabaix, X., Gopikrishnan, P., Plerou, V. and Stanley, H.E. (2003)}
A theory of power-law distributions in financial market fluctuations.
\textit{Nature} \textbf{423}, 267 -- 270.

\hangindent=\parindent\noindent \textsc{Gollier, C. and Pratt, J.W. (1996)}
Weak proper risk aversion and the tempering effect of background risk.
\textit{Econometrica} \textbf{64}(5), 1109 -- 1123.

\hangindent=\parindent\noindent \textsc{Gradshteyn, I.S. and Ryzhik, I.M. (2007)}
\textit{Table of Integrals, Series and Products.} Seventh edition. Academic Press, New York.

%\hangindent=\parindent\noindent \textsc{Jeo, H. (1997)} \textit{Multivariate Models and Dependence Concepts.} CRC Press, Boca Raton.

\hangindent=\parindent\noindent \textsc{Koedijk, K.G., Schafgans, M.M.A. and de Vries, C.G. (1990)}
The tail index of exchange rate returns.
\textit{Journal of International Economics} \textbf{29}(1-2), 93 -- 108.

\hangindent=\parindent\noindent \textsc{Kotz, S., Balakrishnan, N. and Johnson,
N.L. (2000)} \textit{Continuous Multivariate Distributions.} Second edition.  Wiley, New York.

\hangindent=\parindent\noindent \textsc{Longin, F.M. (1996)}
The asymptotic distribution of extreme stock market returns.
\textit{Journal of Business} \textbf{69}(3), 383 -- 408.

\hangindent=\parindent\noindent \textsc{Mathai, A.M. and Moschopoulos, P.G. (1991)}
On a multivariate gamma.
\textit{Journal of Multivariate Analysis} \textbf{39}(1), 135 -- 153.

\hangindent=\parindent\noindent \textsc{Mathai, A.M. and Moschopoulos, P.G. (1992)}
A form of multivariate gamma distribution.
\textit{Annals of the Institute of Statistical Mathematics}
\textbf{44}(1), 97 -- 106.

\hangindent=\parindent\noindent \textsc{Meyers, G.G. (2007)}
The common-shock model for correlated insurance losses.
\textit{Variance} \textbf{1}(1), 40 --52.

\hangindent=\parindent\noindent \textsc{Moschopoulos, P.G. (1985)}
The distribution of the sum of independent gamma random variables.
\textit{Annals of the Institute of Statistical Mathematics}
\textbf{37}, 541 -- 544.

\hangindent=\parindent\noindent \textsc{Pareto, V. (1897)} The new theories of economics. \textit{Journal of Political Economy} \textbf{5}(4), 485 -- 502.

\hangindent=\parindent\noindent \textsc{Pfeifer, D. and Ne\v slehov\' a, J. (2004)}
Modeling and generating dependent risk processes for IRM and DFA. \textit{%
ASTIN Bulletin} \textbf{34}(2), 333 -- 360.

\hangindent=\parindent\noindent \textsc{Pickands, J. (1975)} Statistical
inference using extreme order statistics. \textit{The Annals of Statistics}
\textbf{3}(1), 119 -- 131.

\hangindent=\parindent\noindent \textsc{Ramabhadran, V. (1951)}
A multivariate gamma-type distribution.
\textit{Journal of Multivariate Analysis}
\textbf{38}, 213 -- 232.

\hangindent=\parindent\noindent \textsc{Soprano, A., Crielaard, B., Piacenza, F. and Ruspantini, D. (2009)}  \textit{Measuring Operational and Reputational Risk: A Practitioner's Approach.} Wiley, Chichester.

\hangindent=\parindent\noindent \textsc{Standard \& Poor's  (2015)}
Default, transition and recovery:
2014 annual global corporate default study and rating transitions. Technical report, Standard and Poor's, New York.

\hangindent=\parindent\noindent \textsc{Tsanakas, A. (2008)} Risk
measurement in the presence of background risk. \textit{Insurance:
Mathematics and Economics} \textbf{42}(2), 520 -- 528.

\hangindent=\parindent\noindent \textsc{Vernic, R. (1997)} On the bivariate
generalized Poisson distribution. \textit{ASTIN Bulletin} \textbf{27}(1), 23
-- 32.

\hangindent=\parindent\noindent \textsc{Vernic, R. (2000)} A multivariate
generalization of the generalized Poisson distribution. \textit{ASTIN
Bulletin} \textbf{30}(1), 57 -- 67.

\hangindent=\parindent\noindent \textsc{Vernic, R. (2011)} Tail conditional
expectation for the multivariate
Pareto distribution of the second
kind: Another approach. \textit{Methodology and Computing in
Applied Probability} \textbf{13}(1), 121 -- 137.

\hangindent=\parindent\noindent \textsc{Wang, S. (1996)}
Premium calculation by transforming the layer premium density.
\textit{ASTIN Bulletin} \textbf{26}(1), 71 -- 92.
\appendix

\section{Proofs}

\label{ap} \numberwithin{equation}{section}

\label{appendix-sec}

\begin{proof}[Proof of Proposition \ref{propLTG}]
By construction we readily have that
\begin{eqnarray*}
\hat{G}_{1,\ldots,n}(\mathbf{t})&=&\mathbf{E}\left[e^{-\mathbf{X}'\mathbf{t}}\right]=
\mathbf{E}\left[
\exp\left\{
-\sum_{i=1}^n\sum_{j=1}^{n+1}\frac{c_{i,j}}{\sigma_i}Y_jt_i
\right\}
\right]
=\prod_{j=1}^{n+1}\hat{G}_j\left(
\sum_{i=1}^n \frac{c_{i,j}}{\sigma_i} t_i
\right),
\end{eqnarray*}
which along with (\ref{GamLS}) completes the proof.
\end{proof}

{
\begin{proof}[Proof of Theorem \ref{jointpdf}]
Let $G^{\prod}$ denote a multivariate c.d.f. with mutually
independent gamma-distributed univariate margins.
We have the following string of equations
\begin{eqnarray*}
&&\left(\prod_{i=1}^n \sigma_i\right) f_{1,\ldots,n}(x_1,\ldots,x_n)=
\left(\prod_{i=1}^n \sigma_i\right)(-1)^n \frac{\partial^n}{\partial x_1\cdots \partial x_n} \overline{F}_{1,\ldots,n}(x_1,\ldots,x_n)\\
&=&\mathbf{E}\left[\left(\prod_{i=1}^n \sigma_i\right)(-1)^n \frac{\partial^n}{\partial x_1\cdots \partial x_n}
\exp\left\{
-\sum_{i=1}^n\sum_{j=1}^{n+1}\frac{c_{i,j}}{\sigma_i}Y_jx_i
\right\}
\right] \\
&=&\mathbf{E}\left[\left(\prod_{i=1}^n \sum_{j=1}^{n+1} c_{i,j}Y_j\right)
\exp\left\{
-\sum_{i=1}^n\sum_{j=1}^{n+1}\frac{c_{i,j}}{\sigma_i}Y_jx_i
\right\}
\right] \\
&=&\sum_{\forall i_j\in I}d_c(i_1,\ldots,i_{n+1})  \int_{\mathbf{R}_+^{n+1}}\prod_{j=1}^{n+1}\exp\left\{-y_j\left(\sum_{i=1}^nc_{i,j}\frac{x_i}{\sigma_i} \right) \right\}\frac{\Gamma(\gamma_j+i_j)}{\Gamma(\gamma_j)}dG_{1,\ldots,n+1}^{\Pi}(\boldsymbol{\boldsymbol{y}};\ \boldsymbol{\tilde{\gamma}},1),
\end{eqnarray*}
where $\boldsymbol{\tilde{\gamma}}=(\gamma_1+i_1,\ldots,\gamma_{n+1}+i_{n+1})'$ is
a vector of positive parameters.
The proof is completed by computing the iterated integral.
\end{proof}
}

\begin{proof}[Proof of Theorem \ref{CharLem}]
Let $F_{\boldsymbol{\Xi}}$ denote the c.d.f. of the r.v. $\boldsymbol{\Xi}=(\Xi_1,\ldots,\Xi_n)'$.
The `if' part is immediate from the following obvious relations
\begin{eqnarray}
\overline{F}_{1,\ldots,n}(x_1,\ldots,x_n)
&=&\mathbf{P}[\Lambda_1>\Xi_1x_1,\ldots,\Lambda_n>\Xi_nx_n] \notag \\
&=&\int_{\mathbf{R}_+^n} \mathbf{P}[\Lambda_1>\xi_1x_1,\ldots,\Lambda_n>\xi_nx_n]
dF_{\boldsymbol{\Xi}}(\xi_1,\ldots,\xi_n) \notag \\
&=&\int_{\mathbf{R}_+^n} \exp\left\{-\sum_{i=1}^n \xi_ix_i\right\}
dF_{\boldsymbol{\Xi}}(\xi_1,\ldots,\xi_n)
\label{ProofChar1_1}
\end{eqnarray}
and by Proposition \ref{propLTG}. The `only if' part follows because (\ref{ProofChar1_1})
is the $n$-variate Laplace transform of $Ga_{1,\ldots,n}(\boldsymbol{\gamma}_c^\ast,\
\boldsymbol{\sigma})$, and it is thus unique. This completes the proof.
\end{proof}

{\begin{proof}[Proof of Theorem \ref{minima}]
For the proof, we readily have that
\begin{eqnarray*}
\overline{F}_{-}(x)&=& \overline{F}_{1,\ldots,n}(x,\ldots,x)=\prod_{j=1}^{n+1}
\left(
1+\sum_{i=1}^n \frac{c_{i,j}}{\sigma_i}x
\right)^{-\gamma_j} \\
&=&\int_0^\infty e^{-\lambda x} dF_{Z_1+\cdots+Z_{n+1}}(\lambda),
\textnormal{ where }x\in\mathbf{R}_{+},
\end{eqnarray*}
which establishes the mixture representation.
\end{proof}}

\begin{proof}[Proof of Theorem \ref{ggamma}]
Employing Theorem \ref{minima} with
$Z_j\backsim Ga\left(\gamma_j,\ \left(\sum_{i=1}^n \frac{c_{i,j}}{\sigma_j}\right)^{-1}\right),\ j=1,\ldots,n+1$, Lemma \ref{FL2005}, changing the order of summation and
integration and using equation
(\ref{GamLS}), we have that
\begin{eqnarray*}
\overline{F}_{-}(x)&=&\int_0^\infty \sum_{k=0}^\infty e^{-\lambda x}
 p_k\frac{e^{-\lambda\sigma_{+}}\lambda^{\gamma^\ast+k-1}\alpha_{+}(\boldsymbol{\sigma})^{\gamma^\ast+k}}{\Gamma(\gamma^\ast+k)}d\lambda
 =
 \sum_{k=0}^\infty \left(1+\frac{x}{\alpha_{+}(\boldsymbol{\sigma})}\right)^{-(\gamma^\ast+k)}
 p_k.
\end{eqnarray*}
%\begin{equation*}
%\overline{F}_{-}(x)=\sum_{k=0}^\infty \left(1+\frac{x}{\sigma_{+}}\right)^{-(\gamma^\ast+k)}
% p_k=\mathbf{E}[\Lambda^{-1}]\sum_{k=0}^\infty  \frac{\gamma^\ast+k-1}%{\sigma_{+}}\left(1+\frac{x}{\sigma_{+}}\right)^{-(\gamma^\ast+k-1)-1}
%q_k.
%\end{equation*}
This completes the proof.
\end{proof}

\begin{proof}[Proof of Theorem \ref{cov}]
Let $G^{\prod}$ denote a multivariate c.d.f. with mutually
independent gamma-distributed univariate margins.
We start by employing Lemma \ref{CharLem} and observation (\ref{StRep}), then do change of
variables and obtain that
\begin{eqnarray*}
(\sigma_k\sigma_l)^{-1}\mathbf{E}[X_kX_l]&=&(\sigma_k\sigma_l)^{-1}\mathbf{E}\left[
\frac{1}{\Xi_k}\cdot \frac{1}{\Xi_l}
\right]=\int_{\mathbf{R}_+^3} \frac{1}{(y_3+y_1)(y_3+y_2)}
dG_{1,\ldots,3}^\Pi(\mathbf{y;\gamma,1}) \notag \\
&=&\int_0^\infty \int_0^\infty (1+v)^{-\gamma_{c,l}} (1+u)^{-\gamma_{c,k}}
(1+u+v)^{-\gamma_{c,(k,l)}}du dv \\
&=&\int_0^\infty (1+v)^{-\gamma_{c,l}^\ast} \left(\int_0^\infty  (1+u)^{-\gamma_{c,k}} \left(1+\frac{u}{1+v}\right)^{-\gamma_{c,(k,l)}}du\right) dv  \\
&\overset{(1)}{=}&\int_0^\infty (1+v)^{-\gamma_{c,l}^\ast}
\frac{1}{\gamma_{c,k}^\ast-1}{}_2F_1\left(\gamma_{c,(k,l)},1;\gamma_{c,k}^\ast;\frac{v}{1+v}\right) dv \\
&=&\int_0^1 (1-u)^{(\gamma_{c,l}^\ast-1)-1}
\frac{1}{\gamma_{c,k}^\ast-1}{}_2F_1\left(\gamma_{c,(k,l)},1;\gamma_{c,k}^\ast;u\right) du \\
&\overset{(2)}{=}&
\frac{1}{(\gamma_{c,k}^\ast-1)(\gamma_{c,l}^\ast-1)}{}_3F_2\left(\gamma_{c,(k,l)},1,1;\gamma_{c,k}^\ast,\gamma_{c,l}^\ast;1\right),
\end{eqnarray*}
{
where `$\overset{(1)}{=}$'  holds because of the following integral representation of the Gauss hypergeometric
function (Equation 3.197(5) in Gradshtein and Ryzhik, 2007)
\[
{}_2F_1(\alpha,\beta;\gamma;z)=\frac{\Gamma(\gamma)}{\Gamma(\beta)\Gamma(\gamma
-\beta)}\int_0^\infty t^{\beta-1}(1+t)^{\alpha-\gamma}(1+z t)^{-\alpha}dt
\]
for $\gamma>\beta>0$ and all $z$ such that the integral above converges,
and `$\overset{(2)}{=}$' follows from Equation 7.512(5) in loc. cit. }
This completes the proof.
\end{proof}

\begin{proof}[Proof of Corollary \ref{corrbounds}]
The lower bound follows by setting $\gamma_{c,(k,l)}\equiv 0$ and both of
$\gamma_{c,l}^\ast$ and $\gamma_{c,k}^\ast$ to exceed two. To establish the upper bound,
let $\gamma_{c,k}\rightarrow 0$ and $\gamma_{c,l}\rightarrow 0$ and assume that
$\gamma_{c,(k,l)}$ exceeds two, {then
 we have that
\begin{eqnarray*}
&&\lim_{\gamma_{c,k},\gamma_{c,l}\rightarrow 0}\ _3F_2\left(\gamma_{c,(k,l)},1,1;\gamma_{c,k}+\gamma_{c,(k,l)},\gamma_{c,l}+\gamma_{c,(k,l)};1\right)\\
=&&{}_3F_2\left(\gamma_{c,(k,l)},1,1;\gamma_{c,(k,l)},\gamma_{c,(k,l)};1\right)\\
=&&{}_2F_1\left(1,1;\gamma_{c,(k,l)};1\right)\\
=&&\frac{\gamma_{c,(k,l)}-1}{\gamma_{c,(k,l)}-2},
\end{eqnarray*}
where the last equality holds due to Equation 9.122(1) in Gradshteyn and Ryzhik (2007),}
and the covariance of interest reduces to
\[
\mathbf{Cov}[X_k,\ X_l]\rightarrow
\frac{\sigma_l\sigma_k  }{(\gamma_{c,(k,l)}-1)^2(\gamma_{c,(k,l)}-2)},
\textnormal{ where } 1\leq k\neq l\leq n.
\]
This, along with Proposition \ref{Nt1}, completes the proof.
\end{proof}

\begin{proof}[Proof of Corollary \ref{CL2009Cov}]
To obtain (\ref{covflex1}), we set $c_{i,i}=c_{i,n+1}\equiv 1$ and
zero otherwise. This implies $\gamma_{c,k}^\ast=\gamma_k+\gamma_{n+1}$,
$\gamma_{c,l}^\ast=\gamma_l+\gamma_{n+1}$ and $\gamma_{c,(k,l)}=\gamma_{n+1}$ (see,
Example \ref{ExMMG1991}).
Then the result directly follows from Theorem \ref{cov}.
To establish (\ref{covflex2}), let $c_{i,j}\equiv 1$ for $1\leq j\leq i\leq n $ and zero otherwise.
Then the desired assertion follows
because
\begin{eqnarray*}
\mathbf{Cov}[X_{2,k},\ X_{2,l}]&=&\frac{\sigma_k\sigma_l}{(\gamma_{c,k}^\ast-1)(\gamma_{c,l}^\ast-1)}
({}_{3}F_2(\gamma_{c,(k,l)},1,1;\gamma_{c,(k,l)},\gamma_{c,l}^\ast; 1)-1) \\
&=&\frac{\sigma_k\sigma_l}{(\gamma_{c,k}^\ast-1)(\gamma_{c,l}^\ast-1)}
({}_{2}F_1(1,1;\gamma_{c,l}^\ast; 1)-1) \\
&=&\frac{\sigma_k\sigma_l}{(\gamma_{c,k}^\ast-1)(\gamma_{c,l}^\ast-1)(\gamma_{c,l}^\ast-2)},
\end{eqnarray*}
where the latter equality holds for $\gamma_k^\ast>2$. This completes the proof of the
corollary.
\end{proof}

{
\begin{proof}[Proof of Theorem \ref{cd}]
We first note that
\[
\mathbf{P}[X_k>x_k| X_l=x_l]=\frac{-\frac{\partial}{\partial x_l}\overline{F}_{X_k,X_l}(x_k,x_l)}{f_{X_l}(x_l)}
\textnormal{ for } x_k \textnormal{ and } x_l\textnormal{ both in }\mathbf{R}_+,
\]
and then write
\begin{eqnarray}
%\label{partial_formula}
-\frac{\partial}{\partial x_l}\overline{F}_{X_k,X_l}(x_k,x_l)&=&\left(1+\frac{x_k}{\sigma_k}\right)^{-\gamma_{c,k}}\left[\frac{\gamma
_{c,(k,l)}}{\sigma_l}\left(1+\frac{x_k}{\sigma_k}+\frac{x_l}{\sigma_l} \right)^{-\gamma_{c,(k,l)}-1}\left(1+\frac{x_l}{\sigma_l}\right)^{-\gamma_{c,l}}\right.\nonumber\\
&+&\left. \frac{\gamma_{c,l}}{\sigma_l}\left(1+\frac{x_k}{\sigma_k}+\frac{x_l}{\sigma_l} \right)^{-\gamma
_{c,(k,l)}}\left(1+\frac{x_l}{\sigma_l}\right)^{-\gamma_{c,l}-1}
 \right].
\end{eqnarray}
Plain simplifications complete the proof.
\end{proof}
}

\begin{proof}[Proof of Theorem \ref{cregTh}]
{Note that for the pair $(X_k,\ X_l)',\ 0\leq k\neq l\leq n$, stochastic
representation (\ref{StRep}) is of utmost generality, i.e.,  conditional distribution function (\ref{CondFunction}) coincides with the one of the type I multivariate flexible Pareto model.} Hence the proof is completed by
evoking Theorem 3 of Chiragiev and Landsman (2009) as well as Proposition
\ref{Nt1}.
\end{proof}

\begin{proof}[Proof of Theorem \ref{MinimaCS}]
Note that
\[
\wedge_{j=1}^{n+1} \left(X_{\lambda_j}*\Lambda_j\right)\overset{d}{=}\left(\wedge_{j=1}^{n+1} X_{\lambda_j}
\right)*\boldsymbol{\Lambda},
\]
and hence, for $F_{\boldsymbol{\Lambda}}$ denoting the c.d.f. of the r.v. $\boldsymbol{\Lambda}=(\Lambda_1,\ldots,\Lambda_{n+1})'$,
\begin{eqnarray*}
\overline{F}_{1,\ldots,n}(x_1,\ldots,x_n)
&=&
\int_{\mathbf{R}_+^{n+1}} \prod_{i=1}^{n}\exp\left\{-\sum_{j=1}^{n+1} \frac{c_{i,j} x_i \lambda_j}{\sigma_i}\right\}
dF_{\mathbf{\Lambda}}(\lambda_1,\ldots,\lambda_{n+1}),
\end{eqnarray*}
which, along with \eqref{GamLS}, completes the proof.
\end{proof}

\begin{proof}[Proof of Proposition \ref{CTEmixt}]
 Notice that
\[
\pi_{w}[X]=\frac{\mathbf{E}\left[\mathbf{E}[w(X)|\ \boldsymbol{\Lambda}]\pi_w[X|\ \boldsymbol{\Lambda}]
\right]}{\mathbf{E}[ \mathbf{E}[w(X)|\ \boldsymbol{\Lambda}] ]}
 \textnormal{ for } X\in\mathcal{X},
\]
and set ${w}_1(\boldsymbol{\lambda})=\mathbf{E}[w(X)|\ \boldsymbol{\lambda}]$ and
${v}_1(\boldsymbol{\lambda})=\pi_w[X|\ \boldsymbol{\lambda}]$. This concludes the
proof.
\end{proof}

\begin{proof}[Proof of Corollary \ref{VaRProp}]
Use Proposition \ref{CTEmixt} setting the weight function equal to the Dirac delta function,
or alternatively evoke Proposition \ref{Nt1}. This concludes the proof.
\end{proof}

\begin{proof}[Proof of Corollary \ref{CTEgen}]
Fix $w(x)=\mathbf{1}\{x>VaR_q[X]\}$ for $q\in[0,\ 1)$, the result follows from Proposition
\ref{CTEmixt}.
\end{proof}

\begin{proof}[Proof of Corollary \ref{CTEParXiCor}]
As $X_i|\Lambda=\lambda\backsim Exp(\lambda),\ \lambda>0,\ i=1,\ldots,n$, we readily have that
\[
CTE_{q^\ast}[X_i|\ \Lambda=\lambda]=\frac{1}{\lambda}+VaR_q[X_i]
\]
and
\[
\overline{C}[VaR_q[X_i];\lambda]=e^{-\lambda VaR_{q}[X_i]},
\]
and the assertion holds by Proposition \ref{CTEmixt}. This completes the proof.
\end{proof}

\begin{proof}[Proof of Proposition \ref{CTEparMin}]
We readily have that
\begin{eqnarray*}
CTE_q[X_-]&=&VaR_q[X_-]+\frac{1}{1-q}\int_{VaR_q[X_-]}^\infty \overline{F}_-(x)dx \\
&\overset{(1)}{=}&VaR_q[X_{-}]+\frac{1}{1-q}\int_{VaR_q[X_{-}]}^\infty\left(\sum_{k=0}^\infty \left(1+\frac{x}{\alpha_+(\boldsymbol{\sigma})}\right)^{-(\gamma^\ast+k)}
p_k\right)dx \\
&\overset{(2)}{=}&VaR_q[X_{-}]+\frac{\mathbf{E}[X_-]}{1-q}
\int_{VaR_q[X_{-}]}^\infty\left(\sum_{k=0}^\infty \frac{\gamma^\ast+k-1}{\alpha_+(\boldsymbol{\sigma})}\left(1+\frac{x}{\alpha_+(\boldsymbol{\sigma})}\right)^{-(\gamma^\ast+k-1)-1}
q_k\right)dx ,
\end{eqnarray*}
where `$\overset{(1)}{=}$' follows because of Corollary \ref{ggamma} and `$\overset{(2)}{=}$'
holds since
\[
\mathbf{E}[X_-]=\sum_{k=0}^\infty \frac{\alpha_+(\boldsymbol{\sigma})}{\gamma^\ast+k-1}p_k.
\]
This completes the proof.
\end{proof}

\begin{proof}[Proof of Proposition \ref{ctemaxima}]
We observe that
\[
\int_x^\infty \overline{F}(t)dt=\mathbf{E}[X-x|\ X>x]\overline{F}(x),
\]
for all $x$ in the range of the r.v. $X$. Then the assertion follows
by Proposition \ref{maxrv} and changing the order of summation and integration in
\[
CTE_q[X_+]=\frac{\int_{VaR_q[X_+]}^\infty \sum_{\mathcal{S}\subseteq\{1,\ldots,n\}}
(-1)^{|\mathcal{S}|-1}
\overline{F}_{S-}(x)dx}{1-q}+VaR_q[X_+].
\]
This completes the proof.
\end{proof}

\begin{proof}[Proof of Proposition \ref{econCTE}]
Note that, for $1\leq k \neq l\leq n$,
\begin{eqnarray*}
\mathbf{E}[X_k|\ X_l > VaR_q[X_l]]&=&\int_0^{\infty}
\mathbf{P}[X_k>x|\ X_l>VaR_q[X_l]]
dx\\
&\overset{(1)}{=}&\sigma_k\int_0^\infty
(1+u)^{-\gamma_{c,k}}
\left(1+\frac{u}{1+VaR_q[X_l]/\sigma_l}\right)^{-\gamma_{c,(k,l)}}du,
\end{eqnarray*}
where `$\overset{(1)}{=}$' holds because of observation (\ref{StRep}) and techniques similar to
the ones used in Theorem \ref{cov}. The Proposition then follows by evoking Equation
3.197(5) in Gradshteyn and Ryzhik (2007).
This completes the proof.
\end{proof}
 \end{document}